\documentclass{article}

\usepackage{microtype}
\usepackage{graphicx}
\usepackage{subfigure}
\usepackage{booktabs} 

\usepackage{hyperref}

\usepackage[accepted]{icml2024}

\usepackage{amsmath}
\usepackage{amssymb}
\usepackage{mathtools}
\usepackage{amsthm}

\usepackage[capitalize,noabbrev]{cleveref}

\usepackage{enumitem}
\usepackage{bm, bbm} 

\def \P {\mathbb{P}}   
\def \F {\mathbb{F}}   
\def \Cov {\mathrm{Cov}}   
\def \Var {\mathrm{Var}}   
\def\1{1\!{\rm l}}
\def \W {\mathcal{W}}   
\def \Y {\mathcal{Y}}   
   
\def \R {\mathbb{R}}

\newcommand{\D}{\mathbb{D}}

\newcommand{\mP}{\mathcal{P}}

\newcommand{\mF}{\mathcal{F}}

\newcommand{\bG}{\mathbb{G}}
\newcommand{\indep}{\mathop{\perp\!\!\!\!\perp}}

\newcommand{\Gn}{\mathbb{G}_n}

\newcommand{\Pn}{\mathbb{P}_n}

\newcommand{\Ep}{{\mathrm{E}}}

\renewcommand{\Pr}{{\mathrm{P}}}

\renewcommand{\(}{\left(}
\renewcommand{\)}{\right)}
\renewcommand{\hat}{\widehat}
\newcommand{\En}{{\mathbb{E}_n}}
\renewcommand{\Pr}{{\mathrm{P}}}

\renewcommand{\hat}{\widehat}
\renewcommand{\leq}{\leqslant}
\renewcommand{\geq}{\geqslant}

\usepackage[toc,page,header]{appendix}
\usepackage{minitoc}

\theoremstyle{plain}
\newtheorem{theorem}{Theorem}[section]

\newtheorem{lemma}[theorem]{Lemma}

\theoremstyle{definition}
\newtheorem{definition}[theorem]{Definition}
\newtheorem{assumption}[theorem]{Assumption}
\theoremstyle{remark}
\newtheorem{remark}[theorem]{Remark}

\usepackage[disable,textsize=tiny]{todonotes}

\icmltitlerunning{Distributional Treatment Effects in Randomized Experiments}

\begin{document}

\twocolumn[
\icmltitle{Estimating Distributional Treatment Effects in Randomized Experiments: Machine Learning for Variance Reduction}



\icmlsetsymbol{equal}{*}

\begin{icmlauthorlist}
\icmlauthor{Undral Byambadalai}{comp}
\icmlauthor{Tatsushi Oka}{yyy}
\icmlauthor{Shota Yasui}{comp}
\end{icmlauthorlist}

\icmlaffiliation{yyy}{Department of Economics, Keio University, Tokyo, Japan}
\icmlaffiliation{comp}{CyberAgent, Inc., Tokyo, Japan}

\icmlcorrespondingauthor{Undral Byambadalai}{undral\textunderscore byambadalai@cyberagent.co.jp}
\icmlcorrespondingauthor{Tatsushi Oka}{tatsushi.oka@keio.jp}
\icmlcorrespondingauthor{Yasui Shota}{yasui\textunderscore shota@cyberagent.co.jp}

\icmlkeywords{distributional treatment effects, field experiments, regression adjustment, machine learning}

\vskip 0.3in
]



\printAffiliationsAndNotice{}  

\begin{abstract}
We propose a novel regression adjustment method designed for estimating distributional treatment effect parameters in randomized experiments. 
Randomized experiments have been extensively used to estimate treatment effects in various scientific fields. However, to gain deeper insights, it is essential to estimate distributional treatment effects rather than relying solely on average effects. Our approach incorporates pre-treatment covariates into a distributional regression framework, utilizing machine learning techniques to improve the precision of distributional treatment effect estimators.
The proposed approach can be readily implemented with off-the-shelf machine learning methods and  remains valid as long as the nuisance components are reasonably well estimated. 
Also, we establish the asymptotic properties of the proposed estimator and present a uniformly valid inference method. 
Through simulation results and real data analysis,
we demonstrate the effectiveness of integrating machine learning techniques in reducing the variance of distributional treatment effect estimators in finite samples. 

\end{abstract}

\section{Introduction}\label{sec:intro}
Randomized experiments have played a crucial role in understanding the effects of interventions and guiding policy decisions, ever since the seminal work by \citet{fisher1935design}. The estimation of causal effects through randomized experiments has found widespread application across various scientific disciplines \cite{rubin1974estimating, heckman1997making, imai2005get,imbens2015causal}
and has also become a standard practice within the technology sector \cite{tang2010overlapping, bakshy2014designing, xie2016improving, kohavi2020trustworthy}. 

When analyzing data from randomized experiments, one commonly used measure is the Average Treatment Effect (ATE). However, it is often the case that understanding the distributional treatment effects can provide a richer perspective than solely focusing on overall average effects. Furthermore, while randomized experiments simplify outcome-based analysis, pre-treatment auxiliary information is frequently available. This quest for a more comprehensive understanding of treatment effects, marked by supplementing auxiliary data, calls for new approaches to enhance precision through pre-treatment data incorporation.

In this work, we propose a novel regression-adjustment method to estimate a wide range of distributional parameters in the randomized experiment setup. Our approach draws inspiration from the generic Neyman-orthogonal moment condition \cite{chernozhukov2018debiased}, which facilitates the decoupling of nuisance parameter and treatment effect estimation into two stages.
The nuisance parameters of our interest are the conditional outcome distributions given pre-treatment covariates, and we propose the use of machine learning models (e.g., LASSO, random forests, neural networks, etc.), allowing for complex data and distributional structures.
By integrating these sophisticated machine learning techniques with cross-fitting, we reduce the sensitivity of our treatment effect estimator to errors arising from nuisance parameter estimation. 

Our paper makes several noteworthy contributions.
First, our approach expands the scope of regression adjustment. 
While regression adjustment is commonly employed for variance reduction in the estimation of the ATE with mean regression \cite{deng2013improving, poyarkov2016boosted, guo2021machine}, our method leverages pre-treatment information under the distributional regression framework, incorporating machine learning methods. This enables us to conduct more powerful statistical inference for distributional parameters, including the Distributional Treatment Effect (DTE) and the Quantile Treatment Effect (QTE).
Second, we provide theoretical validation for the regression-adjusted method by demonstrating the variance reduction of the estimator of outcome distribution. Third, we establish asymptotic properties for the proposed treatment effect estimators and provide a uniformly valid inference method. Lastly, our simulation and real-data analysis highlight the significance and effectiveness of our method.

The rest of the paper is structured as follows. Section \ref{sec:related-work} describes related literature. We setup the problem and introduce notations in Section \ref{sec:setup}. Section \ref{sec:reg-adj-estimator} then introduces the regression-adjusted estimators for distributional parameters. We derive the asymptotic results in Section \ref{sec:asymptotic}. Section \ref{sec:empirical} reports empirical results based on simulated experiments and real datasets. Section \ref{sec:conclusion} concludes. The Appendix in the paper includes all proofs, as well as additional experimental details and results. 

\section{Related Work} \label{sec:related-work}
\paragraph{Regression Adjustment} There is an extensive literature investigating the use of pre-treatment covariates to reduce variance in estimating the ATE, dating back to \citet{fisher1932statistical}, followed by \citet{cochran1977sampling, yang2001efficiency, rosenbaum2002covariance, freedman2008regression, freedman2008regression2, tsiatis2008covariate, rosenblum2010simple, lin2013agnostic, berk2013covariance, ding2019decomposing, negi2021revisiting}, among others, in the case of low-dimensional asymptotics. In high-dimensional settings, this topic has been studied by \citet{bloniarz2016lasso, wager2016high, lei2021regression, chiang2023regression}, among others. Recent work by \citet{list2022using} has linked regression adjustment to the semiparametric problem of estimating a low-dimensional parameter when a high-dimensional but orthogonal nuisance parameter is present, focusing on estimating the ATE. Our work extends those existing studies to estimate distributional parameters of treatment effects.

\paragraph{Conditional Average Treatment Effects}
To characterize the heterogeneity in treatment effects, an alternative approach is to condition on observed variables and estimate the Conditional Average Treatment Effect (CATE) \cite{imai2013estimating, athey2016recursive, johansson2016learning, shalit2017estimating, alaa2017bayesian, wager2018estimation, chernozhukov2018debiased, kunzel2019metalearners, shi2019adapting, nie2021quasi, guo2023estimating, sverdrup2023proximal, van2023causal}. 
The CATE can be regarded as the ATE within subgroups defined by observed characteristics, such as gender, age, prior engagement with the platform, and more. Consequently, the CATE captures observed heterogeneity given the information available to the researchers, whereas our approach is valuable for quantifying unobserved heterogeneity and can be extended to estimate distributional parameters conditional on observed information.

\paragraph{Distributional Treatment Effects} Distributional and quantile treatment effects have long been recognized as important parameters to estimate beyond the mean effects. The quantile treatment effect was first introduced by \citet{doksum1974empirical} and \citet{lehmann1975nonparametrics}. Subsequently, estimation and inference methods for distributional and quantile treatment effects in various settings have been developed and applied in econometrics, statistics and machine learning community, including \citet{heckman1997making, imbens1997estimating, abadie2002bootstrap, abadie2002instrumental, chernozhukov2005iv, koenker2005quantile, bitler2006mean, athey2006identification, firpo2007efficient, chernozhukov2013inference, koenker2017handbook, callaway2018quantile, callaway2019quantile, chernozhukov2019generic, ge2020conditional, zhou2022estimating, kallus2024localized}, among others. Some recent works, including \citet{park2021conditional} and \citet{kallus2023robust}, explore the Conditional Distributional Treatment Effects as distributional analysis is useful even after conditioning on observed variables. However, there has been limited research on regression adjustment for unconditional distributional treatment effects. One exception is by \citet{jiang2023regression}, who consider quantile-regression adjustment for the QTE, but under covariate-adaptive randomization. Another exception is the study by \citet{oka2024regression}, which investigates the distribution regression approach using finite-dimensional covariates.
We bridge this gap by proposing regression-adjusted estimators for various distributional parameters when data are obtained from randomized experiments with possibly high-dimensional covariates. Furthermore, our approach accommodates both discrete and mixed discrete-and-continuous outcome distributions, whereas quantile regression adjustment is specifically designed for continuous outcomes.

\paragraph{Semiparametric Estimation} Our work is closely linked to the extensive literature on semiparametric estimation, which addresses the challenge of estimating low-dimensional parameters in the presence of high-dimensional nuisance parameters. This literature includes seminal contributions such as \citet{klaassen1987consistent, robinson1988root, bickel1993efficient,  andrews1994asymptotics, newey1994asymptotic, robins1995semiparametric,  chernozhukov2018debiased, ichimura2022influence}, among others. 
Our setup can be framed as a semiparametric problem characterized by the Neyman-orthogonal moment condition, as outlined in \citet{neyman1959optimal, chernozhukov2018debiased, chernozhukov2022locally}.
Notably, cross-fitting is a commonly used technique in this literature. While our technical arguments share similarities with classical semiparametric methods, our research introduces a novel perspective by emphasizing the significance of flexible machine learning methods for estimating distributional treatment effects, within the framework of randomized experiments. 

\section{Setup and Parameters} \label{sec:setup}
\subsection{Setup and Notation}

We assume that our data are generated from a randomized experiment with $K$ treatment arms. Let $Y \in \mathcal Y \subset \mathbb R$ denote the scalar-valued observed outcome, $W\in \W:= \{1, \dots, K\}$ denote the index of the treatment arm, and $X\in \mathcal X \subset \mathbb{R}^{d_x}$ denote pre-treatment covariates. We observe a size $n$ random sample $\{Z_i\}_{i=1}^{n} := \{(X_i, W_i, Y_i)\}_{i=1}^{n}$ from the distribution of $Z := (X, W, Y)$. 
The probability of assignment to treatment arm $w$ is denoted as $\pi_w:=P(W_i =w)$ satisfying $\sum_{w \in \W} \pi_{w} = 1$, while 
$n_w$ indicates the number of observations in treatment group $w$, satisfying $\sum_{w \in \W} n_{w} = n$. 

We follow the potential outcome framework [e.g., \citet{rubin1974estimating, imbens2015causal}] and let $Y(1), \dots, Y(K)$ denote the potential outcomes, which are hypothetical and represent what the outcome for an individual would be under each treatment scenario. These are unobserved variables and we only observe the outcome for the treatment that is actually administered to each individual. We assume no interference and impose Stable Unit Treatment Values Assumption (SUTVA), which gives us $Y = Y(W)$. In other words, treatment assigned to one unit does not affect the outcome for another unit, and so the potential outcome under any treatment is equal to its observed outcome. Throughout the paper, we also maintain the following two assumptions.

\begin{assumption} \label{ass:independence}
$Y(1), \dots, Y(K), X \indep W.$
\end{assumption}

\begin{assumption} \label{ass:overlap}
$0<\pi_w<1$ for each $w\in\W$.
\end{assumption}

Assumption \ref{ass:independence} states that the treatment indicator is independent of the potential outcomes and the pre-treatment covariates. Assumption \ref{ass:overlap} states that the treatment assignment probabilities are bounded away from 0 and 1. These assumptions are satisfied because we have a randomized experiment where the researcher assigns individuals to treatment groups randomly and have a control over the treatment assignment probabilities.

\subsection{Parameters of interest} 
The parameters of our interest are based on (cumulative) distribution functions of potential outcomes, denoted by 
\begin{align*}
F_{Y(w)}(y):=E[\1_{\{Y(w) \leq y\}}],
\end{align*}
for $y\in\mathcal Y \subset \mathbb R$ and $w\in\W$, where $\1_{\{\cdot\}}$ represents the indicator function.  In general, the potential outcomes $\{Y(w)\}_{w \in \W}$ are unobserved variables. However, under Assumptions \ref{ass:independence} and \ref{ass:overlap}, the potential outcome distribution $F_{Y(w)}(y)$ is the same as the outcome distribution $F_{Y|W}(y|w)$ under each treatment $w$. Therefore, they are identifiable given the data from $Z=(X,W,Y)$.

The results of this paper can be applied to estimate a range of distributional parameters, provided that they rely on (Hadamard) differentiable transformations of potential outcome distributions. We provide a few illustrative examples below.

\paragraph{Example 1: Distributional Treatment Effect} 
Let $w, w' \in \W$ be two different treatment groups. We are interested in the Distributional Treatment Effect (DTE), which is defined as, for $y\in\mathcal Y,$
\begin{align*}
DTE_{w,w'}(y) := F_{Y(w)}(y) - F_{Y(w')}(y). 
\end{align*}
To contrast, the Average Treatment Effect (ATE) is defined as 
\begin{align*}
ATE_{w, w'}: = E[Y(w)] - E[Y(w')].
\end{align*}
The DTE is a parameter that is indexed by a continuum of $y\in\mathcal Y$ and measures the effect of treatment on the whole distribution, whereas the ATE only quantifies the mean effect. As a special case, one can also be interested in the DTE at a certain threshold; i.e., $\mathcal Y$ can be defined to be a singleton set. One advantage of this measure is that it is well-defined for any type of outcome, including discrete, continuous, and mixed discrete-continuous variables.

\paragraph{Example 2: Probability Treatment Effect}
The DTE may not be straightforward to interpret since it measures the differences between two cumulative distributions. However, we can compute more intuitive measures based on these differences. Specifically, the DTE can be used to compute, what we will call, the Probability Treatment Effect (PTE), which is given by
\begin{align*}
PTE_{w, w'}(y, h) := & (F_{Y(w)}(y+h)- F_{Y(w)}(y))\\
& - (F_{Y(w')}(y+h)-F_{Y(w')}(y)),
\end{align*} 
for $y\in\mathcal Y$ and some $h>0.$ The PTE measures the changes in the probability that the outcome falls in interval $(y, y+h]$. The PTE is also well-defined for any type of outcome, including discrete, continuous, and mixed discrete-continuous variables.

\paragraph{Example 3: Quantile Treatment Effect} 
Another common measure used to characterize the entire distribution is the quantile function, defined as 
$F_{Y(w)}^{-1}(\tau):= \inf\{y: F_{Y(w)}(y)\geq \tau\}$ for $\tau \in (0,1)$. The Quantile Treatment Effect (QTE) for quantile $\tau\in(0,1)$ is then given by  
\begin{align*}
    QTE_{w, w'}(\tau) := & F_{Y(w)}^{-1}(\tau) - F_{Y(w')}^{-1}(\tau).
\end{align*} 
The QTE quantifies the difference in quantiles between two potential outcome distributions across a continuum of $\tau \in (0, 1)$. For example, one might be interested in the difference between the medians (when $\tau = 0.5$) of two groups. It is important to note that the QTE is only well-defined for continuous outcomes and may not be appropriate for discrete or mixed discrete-continuous outcomes.

\section{Regression-Adjusted Estimator} \label{sec:reg-adj-estimator}
As explained in the previous section, the potential outcome distribution serves as the fundamental building block for a broad range of distributional parameters.  A simple estimator for the distribution function $F_{Y(w)}(y)$ is the empirical distribution function, given by:
\begin{align*}
    \hat \F^{simple}_{Y(w)}(y) := \frac{1}{n_w} \sum_{i: W_i=w} \1_{\{Y_i \leq y\}},
\end{align*}
for each treatment $w\in\W$. 
While this estimator is an unbiased and consistent estimator, we aim to enhance its precision by leveraging pre-treatment covariates.

To incorporate pre-treatment covariates $X$, we consider the distribution regression framework, in which the conditional distribution function $F_{Y(w)|X}(y|x)$ is regarded as the mean regression for a binary dependent variable 
$\1_{ \{Y(w) \leq y\} }$. That is, for each $y \in \Y$ and $w \in \W$, we can write 
\begin{align*}
    F_{Y(w)|X}(y|X) = E[\1_{\{Y(w)\leq y\}}|X]. 
\end{align*}
For each location $y \in \Y$, 
the conditional mean function can be separetely estimated using various methods, such as linear regression, logistic regression, or other machine learning techniques (e.g., LASSO, random forests, boosted trees, deep neural networks, etc.). Additionally, the distribution regression is applicable for continuous, discrete, and  mixed discrete-and-continuous outcome variable as explained in \citet{chernozhukov2013inference}. 

For the regression-adjusted estimator of $F_{Y(w)}(y)$, the conditional distribution functions are nuisance parameters. We represent them as  
$\gamma_{y}^{(w)}(x):= F_{Y(w)|X}(y|x)$ for each $w \in \W$
and let $\hat \gamma_y^{(w)}(\cdot)$ be an estimator for $\gamma_y^{(w)}(\cdot)$. 
We will explain the necessary conditions for the estimator $\hat{\gamma}_y^{(w)}(\cdot)$ in the following section.
The regression-adjusted estimator of $F_{Y(w)}(y)$ for each $w\in\W$ is 
then defined as 
\begin{align} 
\label{eq:reg-adj-estimator}
    \hat \F_{Y(w)} (y) := & \underbrace{\frac{1}{n_w} \sum_{i: W_i=w} \big (\1_{\{Y_i \leq y\}}-\hat \gamma_y^{(w)}(X_i) \big )}_\text{averaged over observations in treatment group $w$}  \nonumber\\ 
 & + \underbrace{\frac{1}{n} \sum_{i=1}^{n} \hat \gamma_y^{(w)} (X_i).}_\text{averaged over all observations}
\end{align}
The regression-adjusted estimator is obtained by adjusting the empirical distribution function by subtracting $\hat\gamma_y^{(w)}(\cdot)$ that is averaged over observations in that  treatment group and adding $\hat\gamma_y^{(w)}(\cdot)$ that is averaged over all observations. 

This characterization of regression adjustment aligns closely with the concept of the augmented inverse propensity weighted estimator, as explored in \citet{robins1994estimation, robins1995semiparametric}.  
\citet{list2022using} also consider a similar adjustment method for estimating the ATE. We extend this formulation to encompass distribution functions for any arbitrary outcome location $y \in \Y$. It is worth noting that this estimator also serves as an unbiased estimator for the distribution function, as the expected value of the adjustment terms cancels out to zero.

\paragraph{Moment condition problem} 
We rewrite our problem as a moment condition problem. 
In what follows, we will simply write $\gamma_y^{(w)}$ to denote $\gamma_y^{(w)}(\cdot)$ and let
$\gamma_y:= (\gamma_y^{(1)}, \dots, \gamma_y^{(K)})^{\top}$. 
Let $\theta_y^{(w)}:= F_{Y(w)}(y)$ and 
$\theta_y := (\theta_y^{(1)}, \dots, \theta_y^{(K)})^{\top}$. Define moment functions
\begin{align*}
\psi_y(Z; \theta_y, \gamma_y):= 
    \big (\psi_y^{(1)}(Z; \theta_y, \gamma_y), \dots, 
    \psi_y^{(K)}(Z; \theta_y, \gamma_y) \big) ^{\top},    
\end{align*}
where, for each $w \in \W$,  
\begin{align} 
\label{eq:phi-def}
    \psi_y^{(w)}(Z; \theta_y, \gamma_y)
     := & 
    \frac{\1_{\{W=w\}} \cdot (\1_{\{Y\leq y\}}-\gamma_{y}^{(w)}(X))}{\pi_{w}} \notag \\
    & + \gamma_y^{(w)}(X) - \theta_y^{(w)}. 
\end{align}

The following lemma shows what moment conditions are implied by our setup with a randomized experiment. Later, we will show how the regression-adjusted estimator, given in \eqref{eq:reg-adj-estimator}, can be seen as a method of moments estimator that solves the sample counterpart of these moment conditions.

\begin{lemma}[Moment conditions] \label{lemma:moment-condition}
    We have the following moment conditions for a continuum of $y\in\mathcal Y$:
    \begin{align} \label{eq:moment_condition}
    E[\psi_y(Z; \theta_y, \gamma_y)] =0, 
    \end{align}
    where $\psi_y^{(w)}(Z; \theta_y, \gamma_y)$ for each $w\in\W$ is given in \eqref{eq:phi-def}.
    \end{lemma}
Our parameter of interest $\theta_y$ for each $y\in\mathcal Y$ is identified as the solution to the moment condition in \eqref{eq:moment_condition}, where $Z=(X, W, Y)$ is the data and $\gamma_y$ is the possibly infinite-dimensional nuisance parameter.

An important property of the moment conditions defined in \eqref{eq:moment_condition} is that they are Neyman orthogonal with respect to the nuisance parameters \cite{neyman1959optimal, chernozhukov2018debiased, chernozhukov2022locally, ichimura2022influence}. More precisely, the derivative of its expectation with respect to the nuisance parameters vanishes when evaluated at the true parameter values. The following lemma states it formally.

\begin{lemma}[Neyman Orthogonality]\label{lemma:neyman-orthogonality}
For the continuum of moment conditions defined in \eqref{eq:phi-def} for each $w\in\W$ and $y\in\mathcal Y$, we have
\begin{align*}
   \frac{\partial}{\partial t} E[\psi_y(Z; \theta_y, t)] \Big |_{t=\gamma_y} =0, \hspace{0.5cm} a.s.
\end{align*}
\end{lemma}

Neyman orthogonality implies that the moment condition is first-order insensitive to the estimation errors of the nuisance parameters. This property, coupled with a form of sample-splitting called cross-fitting, allows us to derive the asymptotic distribution of the regression-adjusted estimator under mild conditions, even when the conditional distribution functions are estimated via machine learning (ML) methods.

The following lemma shows how the regression-adjusted estimator, given in \eqref{eq:reg-adj-estimator}, can be seen as an estimator that solves the sample counterpart of the moment conditions defined earlier.

\begin{lemma}[Sample moment condition] \label{lemma:sample-moment}
For each $y \in\mathcal Y$, let the vector $\hat \gamma_y$ denote the ML estimator of the vector $\gamma_y$. Then the regression-adjusted estimator $\hat\theta_y:= (\hat F_{Y(1)}(y), \dots, \hat F_{Y(K)}(y))^\top$, where $\hat F_{Y(w)}(y)$ for $w\in\W$ is defined in \eqref{eq:reg-adj-estimator}, is constructed as the solution to the following sample moment condition:
\begin{align*}
   \frac{1}{n}\sum_{i=1}^{n}\psi_y(Z_i; \hat \theta_y, \hat \gamma_y) = 0.
\end{align*}
\end{lemma}

\paragraph{Estimation procedure} We now explain our algorithm, which is summarized in Algorithm \ref{alg:reg-adj-estimator}. Our estimation procedure involves a sample-splitting method called cross-fitting [e.g.,  \citet{chernozhukov2018debiased}]. First, we split the data into $L$ roughly equal-sized folds, where $L>1$. Then, for every observation, we use a ML method and predict nuisance functions $\hat \gamma_y^{(w)}(X_i)$ by training on data from treatment group $w$, excluding data points from the fold the observation belongs to.  This ensures that the observation and the nuisance estimates are independent. Finally, we form a point estimate of $F_{Y(w)}(y)$ by plugging in estimates of the nuisances in \eqref{eq:reg-adj-estimator}. We do this for all $y\in\mathcal Y$ and $w\in\W$. Then we stack the estimators together to get a regression-adjusted estimator $(\hat\theta_y)_{y\in\mathcal Y}$. We discuss the statistical inference in the next section.

\begin{algorithm}[tb]
   \caption{Regression-adjusted estimator}
   \label{alg:reg-adj-estimator}
\begin{algorithmic}
   \STATE {\bfseries Input:} Data $\{(X_i, W_i, Y_i)\}_{i=1}^{n}$ split randomly into $L$ roughly equal-sized folds where $L>1$; $\mathcal{S}$ a supervised learning algorithm

   \FOR{$\ell=1$ {\bfseries to} $L$}
   
   \STATE
   Generate $\hat \gamma_y^{(w)}(X_i)$ predicting $\1_{\{Y_i\leq y\}}$ given $W_i=w$ and $X_i$ for each treatment group $w\in\W$ and each level $y\in \mathcal Y$, by training on data not in fold $\ell$ but in treatment group $w$, using $\mathcal{S}$.
   \ENDFOR
   
   \STATE Compute $\hat F_{Y(w)} (y)$, for each $y\in\mathcal Y$ and $w\in\W$, according to \eqref{eq:reg-adj-estimator}.
   \STATE {\bfseries Result:} Regression-adjusted estimator $(\hat\theta_y)_{y\in\mathcal Y}$.
\end{algorithmic}
\end{algorithm}

\paragraph{Efficiency Gain} 
To illustrate the potential efficiency gain from the proposed regression-adjusted method, consider the scenario where the true conditional distribution function, $\gamma_{y}$, is employed in (\ref{eq:reg-adj-estimator}), leading to the idealized form of the regression-adjusted estimator, denoted by $ \widetilde{\theta}_{y}^{(w)}:= \widetilde{F}_{Y(w)}(y)$. Let $\widetilde{\theta}_{y}:= (\widetilde{\theta}_{y}^{(1)}, \dots,\widetilde{\theta}_{y}^{(K)})^\top$. The following theorem highlights the efficiency improvements of this regression-adjusted estimator in comparison to the empirical distribution function.
As demonstrated in the next section, our estimator asymptotically possess the same efficiency property in terms of variance.

\begin{theorem}
    \label{theorem:efficiency}
    Suppose that $n_{w}/n = \pi_{w} + o(1)$ as $n\to \infty$ for every $w \in \W$. Then, we have \\
    (a)
        for any $w \in \W$ and $y \in \mathcal{Y}$, 
        \begin{align*}
        \Var \big(
        \widehat{\F}_{Y(w)}^{simple}(y)
        \big)
        \ge
        \Var \big(
          \widetilde{\F}_{Y(w)}(y)
        \big) + o(n^{-1}), 
      \end{align*}
       where the equality holds only if $F_{Y(w)|X}(y) = F_{Y(w)}(y)$, \\   
    (b) for any $y \in \Y$,
        \begin{align*}
            \Var
            \big(
            \widehat{\theta}_{y}^{simple}
            \big)
            \succeq
            \Var
            \big(
            \widetilde{\theta}_{y}
            \big) + o(n^{-1}),
        \end{align*}
        where $\succeq$ denotes the positive semi-definiteness. 
        When
        $
               \Var 
     \big (
         F_{Y(w)}(y|X)
         -
         r \cdot 
         F_{Y(w')}(y|X)
         \big )
         > 0 
        $
 for any distinct
$w, w' \in \W$ and $r \in \R$, 
        the positive definite result holds.   

\end{theorem}


Theorem \ref{theorem:efficiency}(a) shows the efficiency gains achieved by applying regression adjustment to distribution functions. Furthermore, Theorem \ref{theorem:efficiency}(b) elaborates on these gains in terms of a vector of regression-adjusted estimators, indicating a marked improvement in the precision of the estimator for the DTE as a special case.

\section{Asymptotic distribution} \label{sec:asymptotic}

In this section, we derive the asymptotic distribution of the regression-adjusted estimator. These results are built upon the functional central limit theorem, functional delta method and other related results from \citet{belloni2017program}.
\paragraph{Additional Notation} We introduce some additional notations to state our results. For a vector $a=(a_1, \dots, a_p)^{\top}\in\mathbb R^p$, $\|a\|= \sqrt{a^{\top}a}$ denotes the Euclidean norm of $a$. Let $\ell^{\infty}(\mathcal Y)$ be the space of uniformly bounded functions mapping an arbitrary index set $\Y$ to the real line; $UC(\Y)$ be the space of uniformly continuous functions mapping  $\Y$ to the real line. $\mathbb G_{n,P} f$ denotes the empirical process $\sqrt{n}\sum_{i=1}^{n}(f(Z_i)-\int f(z)dP(z))$; but we will omit $P$ and simply write $\mathbb G_{n} f$. Let $\mathcal P_n$ denote the set of probability measures, that is weakly increasing in $n$, i.e., $\mathcal P_n \subseteq \mathcal P_{n+1}$. We use $\rightsquigarrow$ to denote the convergence in distribution or law. Lastly, let $\mathbb G_P$ denote the P-Brownian bridge, as defined in Appendix Section \ref{sec:empirical-process}.

The following theorem shows that under regularity conditions, stated fully in Appendix \ref{app:regularity-conditions}, our estimator $(\hat{\theta}_y)_{y\in\mathcal Y}$ is asymptotically Gaussian. Since we employ cross-fitting, the conditions required for the estimation of nuisance functions become much milder compared to not using any data-splitting. It is required that the estimators of nuisance functions attain sufficiently rapid rates of convergence $\tau_n$, in particular $\tau_n = o(n^{-1/4})$ in smooth problems. 

\begin{theorem}[Uniform Functional Central Limit Theorem]\label{theorem:uniform-clt}
Suppose Assumption \ref{ass: S1}, \ref{ass: S2} and \ref{ass: AS} hold. Then, for an estimator $(\hat{\theta}_y)_{y\in\mathcal Y}$ that is defined in Algorithm \ref{alg:reg-adj-estimator}, 
\begin{align*}
    \sqrt{n}(\hat\theta_y -\theta_y)_{y\in\mathcal Y} = (\mathbb G_n  {\psi_y})_{y\in\mathcal Y} +o_p(1)
\end{align*}
in $\ell^{\infty}(\mathcal Y)^K$ uniformly in $P\in\mathcal P_n$, where 
\begin{align*}
Z_{n,P} := (\mathbb G_n  \psi_y)_{y\in\mathcal Y} \rightsquigarrow Z_P:= (\mathbb G_P  \psi_y)_{y\in\mathcal Y} 
\end{align*}
in $\ell^{\infty}(\mathcal Y)^K$ uniformly in $P\in\mathcal P_n$, where the paths of $y \mapsto \mathbb G_p  {\psi_y}$ are a.s. uniformly continuous on a semi-metric space $(\mathcal Y, d_{\mathcal Y})$ and 
\begin{align*}
    & \sup_{P\in\mathcal P_n} E_P \sup_{y\in\mathcal Y} \|\mathbb G_P  {\psi_y}\| < \infty,  \\
    & \lim_{\delta \rightarrow 0} \sup_{P\in\mathcal P_n} E_P \sup_{d_{\mathcal Y}(y,   y') \leq \delta}\|\mathbb G_P  {\psi_y} - \mathbb G_P  {\psi}_{y'} \|=0.
\end{align*}
\end{theorem}

Then, as a special case of the above theorem, for fixed $y\in\Y$, we have pointwise asymptotic normality, stated as $\sqrt{n} (\hat\theta_y -\theta_y) \rightsquigarrow N(0, \Var(\psi_y))$. Note that, $\Var(\psi_y)$ can be consistently estimated via sample moment conditions using cross-fitting as well. The estimate of the asymptotic variance can then be used to construct the confidence intervals in a usual manner.

\paragraph{Functionals of $\theta$}
\begin{figure*}[!htbp]
\vskip 0.2in
\begin{center}
\includegraphics[width=0.8\textwidth]{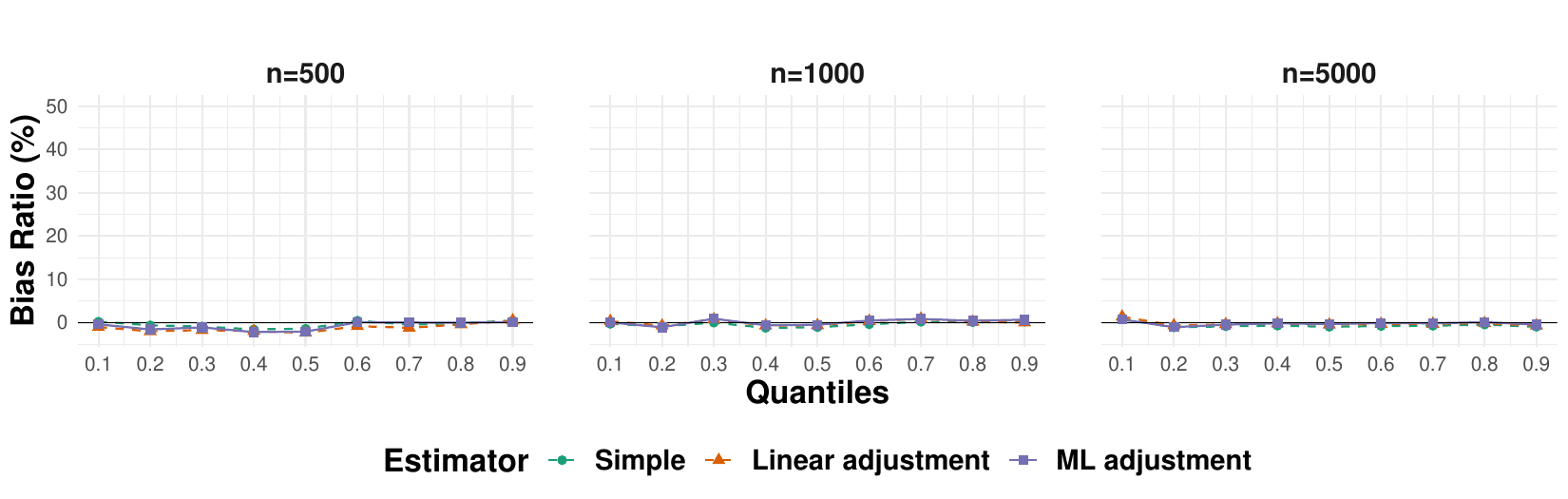}
\includegraphics[width=0.8\textwidth]{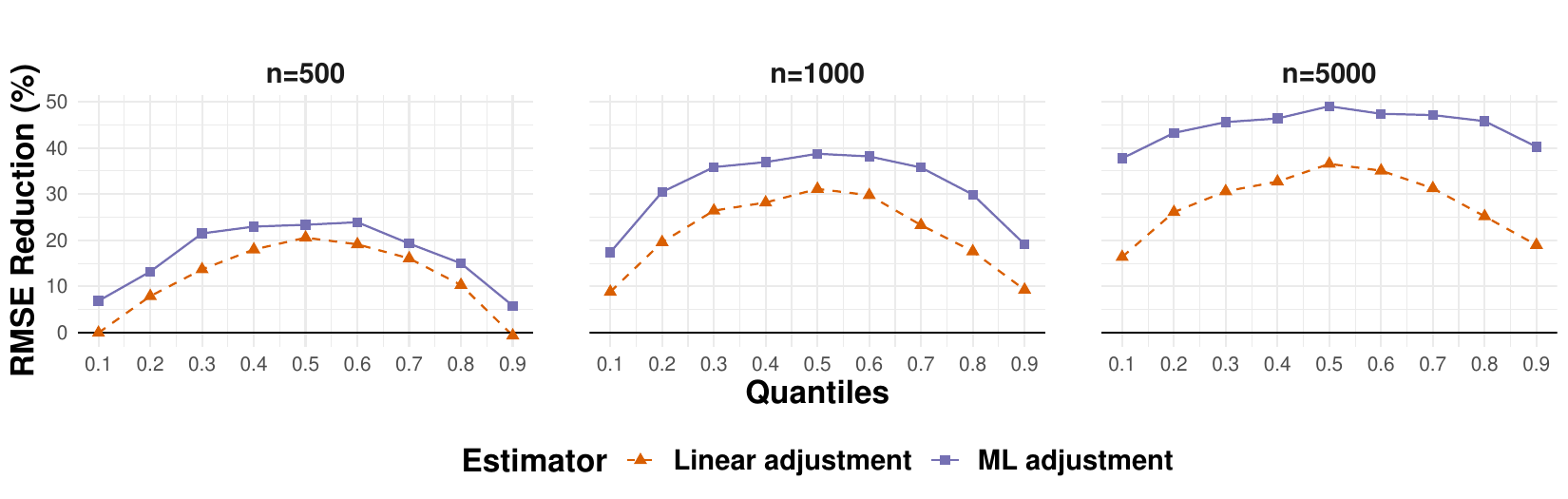}
\caption{Bias (top figure), as a \% of true value, of different DTE estimators and RMSE reduction in \% (bottom figure) of adjusted estimators compared to simple DTE estimator, under sample sizes $\{500, 1000, 5000\}$, calculated over 1000 simulations. The simple estimator is calculated from empirical distribution functions. The regression-adjusted estimators (linear adjustment and ML adjustment based on LASSO) are implemented using 5-fold cross-fitting.}
\label{fig:bias}
\end{center}
\vskip -0.2in
\end{figure*} 

The parameters we are interested in are functionals of potential outcome distributions. The examples include the DTE, the PTE and the QTE we discussed in Section \ref{sec:setup}.  If, for instance, we are interested in the DTE between treatments $1$ and $2$, for each $y\in\mathcal Y$, we can calculate it as $(1,-1, 0, \dots, 0) \theta_y$. Let $\phi(\theta^0)=\phi((\theta_y)_{y\in\mathcal Y})$. The following theorem shows the large sample law of the plug-in estimator $\phi(\hat \theta):= \phi((\hat \theta_y)_{y\in\mathcal Y})$. For complex objects, the inference can be facilitated by bootstrap. The validity of a multiplier bootstrap method \cite{gine1984some} is also shown in the theorem below.

\begin{theorem}[Uniform Limit Theory and Validity of Multiplier Bootstrap for Smooth Functionals of $\theta$] \label{thm:uniform-bootstrap} Suppose that for each $P \in  \mP:= \cup_{n \geq n_0} \mP_n$,  $\theta^0= \theta^0_P$ is an element of a compact set
 $\mathbb{D}_{\theta}$.
Let $\phi: \D_{\phi} \subset\ell ^{\infty}(\mathcal Y)^{K} \longmapsto \ell^{\infty}(\mathcal Q)$ be Hadamard-differentiable uniformly in $ \theta \in \D_{\theta} \subset \D_{\phi}$ tangentially to $UC(\mathcal Y)^{K}$ with derivative map  $\phi_{\theta}^{\prime}$. Then,
\begin{align*}
\sqrt{n}(\phi(\hat \theta) - \phi(\theta^0)) \rightsquigarrow T_P:= \phi'_{\theta^0_P}(Z_P) \text{ in } \ell^{\infty}(\mathcal Q), 
\end{align*}
uniformly in $P\in\mathcal P_n$, where $T_P$ is a zero mean tight Gaussian process, for each $P\in\mathcal P$. Moreover, 
\begin{align*}
 \sqrt{n}(\phi(\hat\theta^*) - \phi(\hat \theta)) \rightsquigarrow_B T_P \text{ in } \ell^{\infty}(\mathcal Q), 
\end{align*}
uniformly in $P\in\mathcal P_n$.
\end{theorem}
Here $\rightsquigarrow_B$ denotes weak convergence of the bootstrap law in probability, as defined in Appendix \ref{app:asymptotics}. $\phi(\hat\theta^*)=\phi((\hat \theta^*)_{y\in\mathcal Y})$ is the bootstrap version of $\phi(\hat \theta) $, and $\hat\theta_y^*=\hat\theta_y + n^{-1} \sum_{i=1}^{n} \xi_i\psi_y(Z_i; \hat\theta_y, \hat\gamma_y)$ is the multiplier bootstrap version of $\hat\theta_y$. More details about the multiplier bootstrap procedure to obtain pointwise and uniform confidence bands can be found in the Appendix \ref{app:multiplier-bootstrap}. The assumption of Hadamard differentiability is imposed so that we can use the delta method. The formal definition can be found in the Appendix \ref{app:asymptotics}.

\section{Empirical results} \label{sec:empirical}
In this section, we compare our regression-adjusted estimators to simple estimators in two types of experiments. In the first experiment, we use a synthetic dataset to assess the performance of our proposed method in finite samples. For the second experiment, we reanalyze data from a randomized experiment, conducted by \citet{ferraro2013using}, to compare the methods using real-world data.

\subsection{Simulation Study} \label{subsec:simulation}
\begin{figure*}[!htbp]
\vskip 0.2in
\begin{center}
\includegraphics[width=0.8\columnwidth]{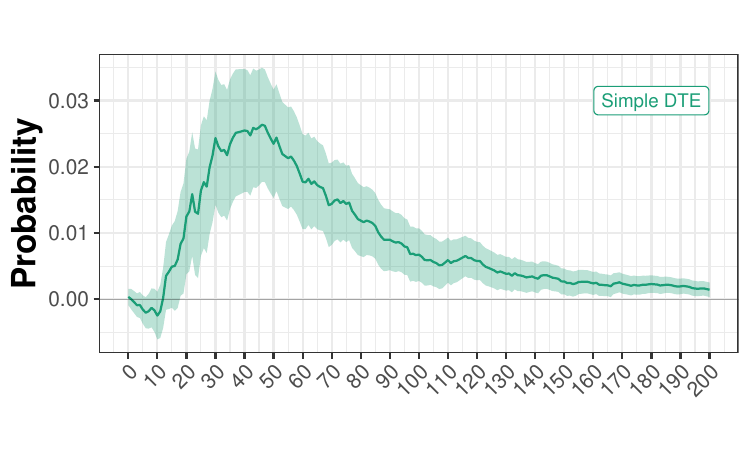}
\includegraphics[width=0.8\columnwidth]{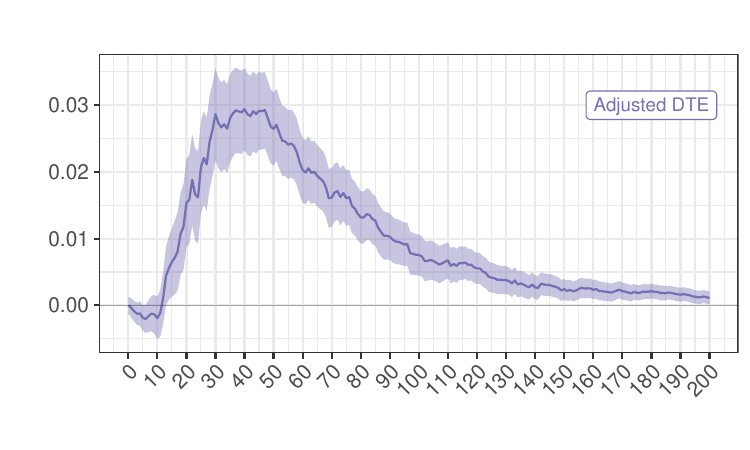}
\includegraphics[width=0.8\columnwidth]{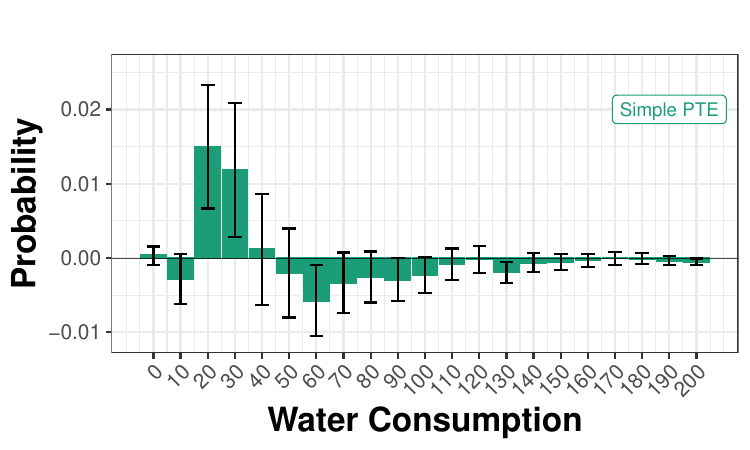}
\includegraphics[width=0.8\columnwidth]{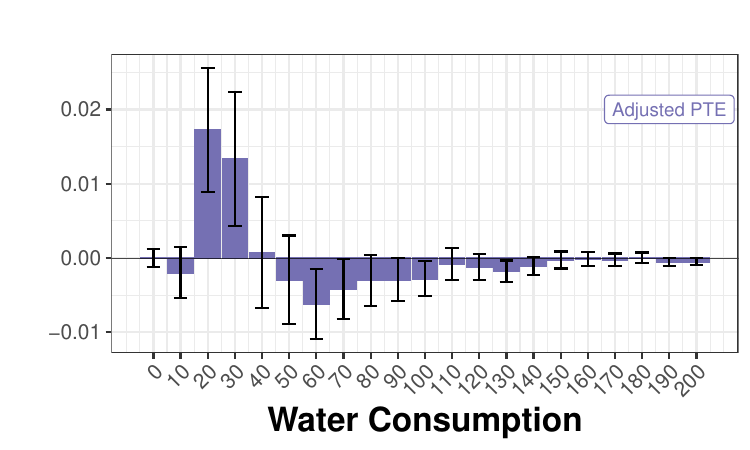}
\caption{Distributional Treatment Effect (DTE) and Probability Treatment Effect (PTE) of a nudge (Strong Social Norm vs. Control) on water consumption (in thousands of gallons). The top left figure  represents the simple DTE; the top right figure depicts the regression-adjusted DTE, computed for $y\in\{0,1,2,\dots, 200\}$. The bottom left figure represents the simple PTE; the bottom right figure represents the regression-adjusted PTE, computed for $y\in\{0,10,20,\dots, 200\}$ and $h=10$. The regression adjustment is implemented via gradient boosting with 10-fold cross-fitting. The shaded areas and error bars represent the 95\% pointwise confidence intervals. $n=78,500$.}
\label{fig:water-dte}
\end{center}
\vskip -0.2in
\end{figure*}
We conduct Monte Carlo simulation study to evaluate the performance of our adjusted estimators in finite samples. We compare the simple estimator with two regression-adjusted estimators - linear adjustment and ML adjustment. The simple estimator is based on empirical distribution functions. The regression-adjusted estimators are calculated according to the procedure in Algorithm \ref{alg:reg-adj-estimator} with 5 folds. For the linear adjustment, we use a linear regression for estimating $\hat \gamma_y$. For the ML adjustment, we use logistic LASSO for estimating $\hat \gamma_y$.

In this experiment, we generate i.i.d. sample of size $n \in \{500, 1000, 5000\}$ with covariates, a binary treatment and a continuous outcome. We design the data generating process such that the half of covariates are irrelevant to the outcome. Appendix \ref{sec:app_experiment} contains more details about the data generating process and describes the evaluation metrics. 

The top figure of Figure \ref{fig:bias} plots the bias of these estimators as a \% of the true values of the DTE, across different quantiles under the sample sizes we consider. We confirm that the bias is small for all estimators across all quantiles. Even when sample size is small ($n=500$), the bias is at most $2\%$. This is as expected since all estimators are unbiased estimators of the distribution functions and hence the DTE.

Next, we turn to the RMSE. The bottom figure of Figure \ref{fig:bias} plots the RMSE reduction in $\%$ terms for the linear and ML adjustment, compared to the simple estimator. We see that the linear adjustment and ML adjustment estimators yield smaller RMSE compared to the simple estimator across all quantiles. Moreover, we see that ML adjustment outperforms linear adjustment in all cases. Specifically, for $n=5000$, RMSE reduction over the simple estimator is around 40-50\% for the ML adjusted estimator, while it is between roughly 15-35\% for the linearly adjusted estimator. This is as expected because our data generating process consists of variables that are irrelevant to the outcome, and so ML adjustment better captures the relationship between the outcome and covariates compared to the linear model. Improved prediction quality for $\hat \gamma_y$ results in more variance reduction for the DTE.

\subsection{Nudges to reduce water consumption}
We reanalyze data from a randomized experiment conducted by \citet{ferraro2013using} in 2007, to examine the effect of norm-based messages or nudges on water usage in Cobb County, Atlanta, Georgia. Three different interventions to reduce water usage were implemented and compared to the control group (no nudge). 
\citet{list2022using} re-estimate the regression-adjusted ATE for the intervention called ``strong social norm'' that combines prosocial appeal and social comparison (the strongest intervention) relative to the control group. We extend their analysis by estimating the regression-adjusted DTE and PTE of this intervention over the control group ($W\in\{0,1\}$) using the same pre-treatment covariates $X$, which is monthly water consumption in the year prior to the experiment. Thus, the dimension of the covariate space $\mathcal X$ is $d_x=12$. The outcome variable $Y$ is level of water consumption from June to September of 2007. The unit of our outcome variable is in thousands of gallons and is discretely distributed.
Note that although the measure in gallons appears to be approximately continuous in practice, the presence of subtle discreteness can create problems for both theoretical and practical statistical inference. Thus, the QTE is not  applicable here. The results for the other two treatments - ``technical advice'' and ``weak social norm'' - relative to the control group are summarized in Appendix \ref{sec:water-nudge}.

Figure \ref{fig:water-dte} plots the DTE and the PTE of the intervention compared to the control group. We compute the DTE for $y\in\{0,1,2,\dots, 200\}$. Figure \ref{fig:water-dte} top left represents the simple estimate of the DTE, whereas the top right figure depicts the regression-adjusted estimate of the DTE. For regression adjustment, we estimate the conditional distribution functions $\hat \gamma_y$ via gradient boosting using 10-fold cross-fitting. The shaded areas represent the 95\% pointwise confidence intervals. We can see that the regression-adjusted DTE has substantially smaller variance and hence tighter confidence intervals compared to the simple DTE, especially between 15 and 110. Based on the regression-adjusted DTE results, we see that the DTE is increasing up until 40-50 and starts declining after that. We can draw conclusions about how the outcome distributions differ under treatment and control, keeping in mind this is differences in cumulative distributions.

More intuitive measure of this distribution change is the PTE. Figure \ref{fig:water-dte} bottom left represents the simple PTE, while the bottom right represents the regression-adjusted PTE, with the 95\% pointwise confidence intervals. We show the PTE in increments of $h=10$ for $y\in\{0, 10, 20,\dots, 200\}$. We see that the treatment is effective in that it reduces the probability of high water consumption and increases the probability of low water consumption. Specifically, the results from the regression-adjusted PTE indicate increase in water usage in the range of 20-40; decrease in water usage in the range of 60-110 (with a minor exception that within the range (80, 90], which is represented by point 80 in the graph, the confidence interval exceeds 0 by only a little). The variance reduction is especially helpful at the range (70, 80]. The probability change for the range (70,80] is not significant under simple estimates, while it is significantly negative under regression adjustment.

\section{Conclusion} \label{sec:conclusion}
We provide a novel regression adjustment method to estimate various measures of distributional treatment effects to capture heterogeneity. Our framework accommodates high-dimensional setup with many pre-treatment covariates and offers flexible modeling by incorporating machine learning techniques for the regression adjustment. 

Some limitations of our method are as follows. Firstly, we consider a setting where we have an experimental data with perfect compliance and no interference (no network or peer effects). While suitable for some applications, these assumptions may prove restrictive in other contexts. Secondly, our approach relies on the presence of pre-treatment covariates highly predictive of the outcome. While we enhance variance reduction compared to linear regression by employing flexible machine learning methods to improve prediction quality, substantial variance reduction may not occur if the covariates lack high-quality information. Thirdly, we focus on a setup where experimental data is already collected, neglecting opportunities to incorporate variance reduction strategies at the design stage of the experiment. These limitations suggest avenues for future research on distributional analysis that incorporates these concerns.

\section*{Acknowledgements}
We extend our gratitude to the four anonymous reviewers and the program chairs for their insightful comments and discussions, which significantly enhanced the quality of this paper. We also appreciate the feedback provided by Hiroki Yanagisawa and Shuting Wu. Additionally, Oka acknowledges the financial support from JSPS KAKENHI Grant Number 24K04821.

\section*{Impact Statement}
This paper presents work whose goal is to advance the field of Machine Learning. There are many potential societal consequences of our work, none which we feel must be specifically highlighted here.

\clearpage
\bibliography{dte.bib}
\bibliographystyle{icml2024}

\newpage 
\onecolumn
\appendix

{\Large \textbf{Appendix} }

The Appendix is structured as follows. Section \ref{app:notation} summarizes the notation in the main text and introduces additional notations used in the Appendix. Section \ref{app:moment} provides key tools and proofs for the claims appeared in Section \ref{sec:reg-adj-estimator}. Section \ref{app:multiplier-bootstrap} describes the multiplier bootstrap procedure for inference. Section \ref{app:asymptotics} provides key tools and proofs for the claims presented in Section \ref{sec:asymptotic}. Finally, Section \ref{sec:app_experiment} contains more detailed information and additional results from the experiments.


\section{Summary of Notation} \label{app:notation}
\begin{table}[htbp]    \centering
    \caption{Summary of Notation}
    \begin{tabular}{ll}
    \hline \hline
    \multicolumn{2}{l}{\textbf{Notation in the main text}}\\
    $X$ & Pre-treatment covariates \\
    $W$ & Treatment variable \\
    $Y$ & Outcome variable \\
    $Y(w)$ & Potential outcome for treatment group $w$ \\
    $\{Z_i\}_{i=1}^{n}$ & $\{(X_i, W_i, Y_i)\}_{i=1}^{n}$, observed data \\
    $\pi_w$ & Treatment assignment probability for treatment group $w$ \\
    $n_w$ & Number of observations in treatment group $w$ \\
    $F_{Y(w)}(y)$ & $E[\1_{\{Y(w) \leq y\}}]$, potential outcome distribution function \\
    $\theta_y$ & $(F_{Y(1)}(y), \dots, F_{Y(K)}(y))^{\top}$, vector of potential outcome distribution functions \\
    $\gamma_y^{(w)}(x)$ & $E[\1_{\{Y\leq y\}}|W=w, X=x]$, conditional distribution function \\
    $\gamma_y$ & $(\gamma_y^{(1)}, \dots, \gamma_y^{(K)})^\top$, vector of conditional distribution functions \\
    $\psi_y^{(w)}(Z; \theta_y, \gamma_y)$ & moment function defined in \eqref{eq:phi-def} \\
    $\psi_y(Z)$ & $\psi_y(Z; \theta_y, \gamma_y) :=
    \big (\psi_y^{(1)}(Z; \theta_y, \gamma_y), \dots, 
    \psi_y^{(K)}(Z; \theta_y, \gamma_y) \big) ^{\top}$, vector of moment functions \\
    $\succeq$ & positive semi-definiteness \\
    $\|a\|$ & $\sqrt{a^{\top}a}$, Euclidean norm of a vector $a=(a_1, \dots, a_p)^{\top}\in\mathbb R^p$ \\
    $\ell^{\infty}(\mathcal Y)$ & space of uniformly bounded functions mapping an arbitrary index set $\Y$ to the real line \\
    $UC(\Y)$ & space of uniformly continuous functions mapping an arbitrary index set  $\Y$ to the real line \\
    $\mathbb G_{n}f $ & $\sqrt{n}\sum_{i=1}^{n}(f(Z_i)-\int f(z)dP(z))$, empirical process  \\
    $\mathbb G_P$ & P-Brownian bridge \\
    $Z_{n,P}$  & $(\mathbb G_n  \psi_y)_{y\in\mathcal Y}$ \\
    $Z_P$ & $(\mathbb G_P  \psi_y)_{y\in\mathcal Y}$ \\
    $T_P$ & $\phi'_{\theta^0_P}(Z_P)$, where $\phi'_{\theta}$ is derivative map of functional $\phi$ \\
    $\mathcal P_n$ & set of probability measures, that is weakly increasing in $n$ \\
    $\rightsquigarrow$ & convergence in distribution or law \\
    $\rightsquigarrow_B$  & weak convergence of the bootstrap law in probability \\
    &\\
    \multicolumn{2}{l}{\textbf{Additional notation in the Appendix}} \\
    $N(\epsilon,\mathcal{F}, \| \cdot \|)$ & $\epsilon$-covering number of the class of functions $\mathcal{F}$ with respect to the norm $\| \cdot \|$ \\
    $a\vee b$ &  $\max\{a,b\}$ for real numbers $a$ and $b$ \\
    $a \wedge b $ & $\min\{a,b\}$ for real numbers $a$ and $b$ \\
    $[K]$ & $\{1, \dots, K\}$ for a positive integer $K$ \\
    $x_n \lesssim y_n$  & for sequences $x_n$ and $y_n$ in $\mathbb R$, $x_n \leq Ay_n$ for a constant $A$ that does not depend on $n$\\
    $\|\cdot\|_{P,q}$ & $L^q(P)$ norm \\
    $BL_1(\mathbb{D})$ & space of functions mapping $\mathbb{D}$ to $[0,1]$ with Lipschitz norm at most 1 \\
    
  $X_n= O_P(a_n)$ & $\lim_{K\to\infty}\lim_{n\to\infty} P(|X_n| > K a_n)=0$ for a sequence of positive constants $a_n$\\
  $X_n =o_P(a_n)$ & $\sup_{K>0}\lim_{n\to\infty} P(|X_n|>K a_n)=0$ for a sequence of positive constants $a_n$ \\
    
    \hline \hline
    \end{tabular}
    \label{tab:notation_summary}
\end{table}

\section{Key Tools and Proofs for Section \ref{sec:reg-adj-estimator}} \label{app:moment}

\subsection{Proofs of Lemmas in Section \ref{sec:reg-adj-estimator}}
\begin{proof}[\textbf{Proof of Lemma \ref{lemma:moment-condition}}]
Let $y\in\mathcal Y$ held constant. It is sufficient to show that each element in the vector of moment conditions  in \eqref{eq:moment_condition} holds. 
From the definition of $\psi_y^{(w)}(Z; \theta_y, \gamma_y)$ in \eqref{eq:phi-def}, we can show, for each $w\in\W$,  
\begin{align*}
 E\big[\psi_y^{(w)}(Z; \theta_y, \gamma_y)  \big ] 
& = E\Big[\frac{\1_{\{W=w\}} \cdot (\1_{\{Y\leq y\}}-\gamma_y^{(w)}(X))}{\pi_w} + \gamma_y^{(w)}(X)- F_{Y(w)}(y) \Big]\\
& =E\Big[\frac{\1_{\{W=w\}} \cdot (\1_{\{Y(w)\leq y\}}-\gamma_y^{(w)}(X))}{\pi_w} + \gamma_y^{(w)}(X)-F_{Y(w)}(y) \Big] \\
& = E\Big[
E[\1_{\{W=w\}}] \cdot \frac{(\1_{\{Y(w)\leq y\}}-\gamma_y^{(w)}(X))}{\pi_w} + \gamma_y^{(w)}(X)-F_{Y(w)}(y) \Big] \\
& = E\Big[\pi_w \cdot
\frac{(\1_{\{Y(w)\leq y\}}-\gamma_y^{(w)}(X))}{\pi_w}  + \gamma_y^{(w)}(X)-F_{Y(w)}(y) \Big]\\
& = E\Big[\1_{\{Y(w)\leq y\}}-\gamma_y^{(w)}(X)+ \gamma_y^{(w)}(X)-F_{Y(w)}(y) \Big] \\
& = 0.
\end{align*}
Here, the first equality is due to 
the definition in \eqref{eq:phi-def}
and 
the second equality comes from the definition of potential outcomes $Y=Y(W)$. The third equality holds because of the independence assumption in Assumption \ref{ass:independence}. The fourth equality comes from the fact that $E[\1_{\{W=w\}}]=P(W=w)=\pi_w$. Finally, all terms cancel out to be zero since $E[\1_{\{Y(w)\leq y\}}]=F_{Y(w)}(y)$ by definition.
\end{proof}

\begin{proof}[\textbf{Proof of Lemma \ref{lemma:neyman-orthogonality}}]
For each $y\in\mathcal Y$ and $w\in\W$, with probability approaching 1, we have
\begin{align*}
\frac{\partial}{\partial t} 
E\big[
\psi_y^{(w)}(Z; \theta_y, t)  \big] & 
=
E\Big[\frac{\partial}{\partial t} \Big(\frac{\1_{\{W=w\}} \cdot (\1_{\{Y\leq y\}}-t)}{\pi_w} +t-F_{Y(w)}(y)\Big) \Big] \\
 & = E\Big[-\frac{\1_{\{W=w\}}}{\pi_w} +1 \Big] \\
 & = -\frac{E[\{\1_{\{W=w\}}]}{\pi_w} +1 \\
 & = 0.
\end{align*}
Thus, the desired conclusion follows. 
\end{proof}

\vspace{0.1cm}
\begin{proof}[\textbf{Proof of Lemma \ref{lemma:sample-moment}}]
For each  $y\in\mathcal Y$ and $w\in\W$, we substitute the estimate 
$\hat\gamma_y^{(w)}(X_i)$ for $\gamma_y^{(w)}(X_{i})$ and $n_w/n$ for $\pi_w$ in equation \eqref{eq:moment_condition}. Then, 
$
\widehat \theta_{y}:= 
\big(
\widehat{\F}_{Y(1)}(y) , 
\dots, 
\widehat{\F}_{Y(K)}(y) 
\big)^{\top}
$
solves the sample counterpart of equation \eqref{eq:moment_condition} for $\theta_{y}$, so that 
\begin{align*}
\hat{\F}_{Y(w)}(y) = & \frac{1}{n}\sum_{i=1}^{n} \Big[\frac{\1_{\{W_i=w\}} \cdot (\1_{\{Y_i\leq y\}}-\hat \gamma_y^{(w)}(X_i))}{n_w/n} + \hat \gamma_y^{(w)}(X_i)\Big].
\end{align*}
Rearranging the terms in the above equation, we obtain
\begin{align*}
\hat \F_{Y(w)}(y) = &  \frac{1}{n_w}\sum_{i:W_i=w}\big(\1_{\{Y_i\leq y\}}-\hat \gamma_y^{(w)}(X_i)\big) + \frac{1}{n}\sum_{i=1}^{n} \hat \gamma_y^{(w)}(X_i), 
\end{align*}
which is the regression-adjusted estimator given in equation (\ref{eq:reg-adj-estimator}). 
\end{proof}

\subsection{Proof of Theorem \ref{theorem:efficiency}}

For the sake of proof completeness, we first present a variant of Lagrange's identity and Bergstr\"{o}m's inequality in the below lemma, which is useful to prove the efficiency gain of the regression adjustment. 

\vspace{0.2cm}
\begin{lemma}
    \label{lemma:L-ind}
    For any 
    $(a_{1}, \dots, a_{K}) \in \R^{K}$
    and 
    $(b_{1}, \dots, b_{K}) \in \R^{K}$
    with $b_{k} > 0 $ for all $k=1, \dots, K$,
    we can show that 
    \begin{align*}
     \sum_{k=1}^{K} \frac{a_{k}^{2}}{b_{k}}   
     - 
     \frac{
     \big (
     \sum_{k=1}^{K} a_{k}
     \big )^2 
    }{
     \sum_{k=1}^{K} b_{k}
    }
    = 
     \frac{
     1
    }{
     \sum_{k=1}^{K} b_{k}
    }
    \cdot 
    \frac{1}{2}
    \sum_{k=1}^{K}
    \sum_{\substack{\ell =1 \\ \ell \neq k}}^K 
    \frac{
    (a_{k} b_{\ell} 
    - 
    a_{\ell} b_{k}
    )^2      
    }{
    b_{k} b_{\ell}
    } ,
    \end{align*}
    which implies 
    Bergstr\"{o}m's inequality, given by 
  \begin{align*}
     \sum_{k=1}^{K} \frac{a_{k}^{2}}{b_{k}}   
     \geq 
     \frac{
     \big (
     \sum_{k=1}^{K} a_{k}
     \big )^2 
    }{ 
     \sum_{k=1}^{K} b_{k}
    }
    . 
    \end{align*}

\end{lemma}
\begin{proof}
    Lagrange's identity is that, 
    for any 
    $(c_{1}, \dots, c_{K}) \in \R^{K}$
    and 
    $(d_{1}, \dots, d_{K}) \in \R^{K}$, 
    \begin{align}
    \label{eq:L-ind}
        \bigg (\sum_{k=1}^{K} c_{k}^2
        \bigg)
        \bigg(\sum_{k=1}^{K} d_{k}^2
        \bigg)
        - 
        \bigg(\sum_{k=1}^{K} c_{k}d_{k}
        \bigg)^2
        =
    \frac{1}{2}
    \sum_{k=1}^{K}
    \sum_{\substack{\ell =1 \\ \ell \neq k}}^K 
     (c_{k} d_{\ell} 
    - 
    c_{\ell} d_{k}
    )^2   .
    \end{align}
    Fix arbitrary 
    $(a_{1}, \dots, a_{K}) \in \R^{K}$
    and 
    $(b_{1}, \dots, b_{K}) \in \R^{K}$
    with $b_{k} > 0 $ for all $k=1, \dots, K$. 
    Then, 
    taking 
    $c_{k} = a_{k} / \sqrt{b_{k}}$ 
    and 
    $d_{k} = \sqrt{b_{k}}$
    for all $k=1, \dots, K$
    in (\ref{eq:L-ind}), we can show that 
    \begin{align*}
        \bigg (
        \sum_{k=1}^{K} 
        \frac{a_{k}^2}{b_{k}}
        \bigg)
        \bigg(\sum_{k=1}^{K} b_{k}
        \bigg)
        - 
        \bigg(\sum_{k=1}^{K} a_{k}
        \bigg)^2
     &   =
    \frac{1}{2}
    \sum_{k=1}^{K}
    \sum_{\substack{\ell =1 \\ \ell \neq k}}^K 
     \bigg (
     \frac{a_{k}}{ \sqrt{b_{k}}}
     \sqrt{b_{\ell}} 
     - 
     \frac{a_{\ell}}{ \sqrt{b_{\ell}}}
     \sqrt{ b_{k}} 
     \bigg )^2   \\ 
     & = 
         \frac{1}{2}
    \sum_{k=1}^{K}
    \sum_{\substack{\ell =1 \\ \ell \neq k}}^K 
    \frac{
    (a_{k} b_{\ell} 
    - 
    a_{\ell} b_{k}
    )^2      
    }{
    b_{k} b_{\ell}
    },
    \end{align*}
    which leads to the desired equality. 
    Also, the last expression in the math display above 
    is non-negative, which leads to Bergstr\"{o}m's inequality. 
\end{proof}
\vspace{0.5cm}

To prove Theorem \ref{theorem:efficiency}, we introduce additional notation. 
Define the empirical probability measures of $X$ as 
\begin{align*}
    \widehat{\P}_{X} :=   
    \frac{1}{n}
    \sum_{i=1}^{n} 
    \delta_{X_i}
    \ \ \mathrm{and} \ \ 
    \widehat{\P}_{X}^{(w)} :=   
    \frac{1}{n_{w}}
    \sum_{i=1}^{n} 
    \1_{\{ W_{i} = w \} } \cdot 
    \delta_{X_i},
\end{align*}
for all observations and observations in the treatment group $w \in \W$, respectively. 
Here, $\delta_{x}$ is the measure that assigns mass 1 at $x \in \mathcal{X}$
and thus $\widehat{\P}_{X}$ and $\widehat{\P}_{X}^{(w)}$ can be interpreted as the random discrete
probability measures, which put mass $1/n$  and $1/n_{w}$ at each of the $n$ and $n_{w}$ points $\{X_{i}\}_{i=1}^{n}$
and 
$\{X_{i}:W_{i} =w\}_{i=1}^{n}$, respectively. 
Given a real-valued function $f: \mathcal{X} \to \R$, we denote by 
\begin{align*}
    \widehat{\P}_{X} f 
    = \int f \widehat{\P}_{X} 
    = \frac{1}{n} \sum_{i=1}^{n}f(X_{i}) 
    \ \ \ \ \mathrm{and} \ \ \ \
    \widehat{\P}_{X}^{(w)} f 
    = \int f \widehat{\P}_{X}^{(w)} 
    = \frac{1}{n_{w}} \sum_{i=1}^{n} \1_{\{ W_{i} = w \} } \cdot f(X_{i}) . 
\end{align*}

Given that the true conditional distribution $\gamma_{y}^{(w)}(X) \equiv F_{Y(w)|X}(y|X)$, 
the infeasible version of regression-adjusted distribution function for treatment $w \in \mathcal{W}$ is written as
\begin{align*}
  \widetilde{\F}_{Y(w)}(y)
  = 
  \widehat{\F}_{Y(w)}^{simple}(y)
  - 
  (\widehat{\P}_{X}^{(w)} - \widehat{\P}_{X})
  \gamma_{y}^{(w)}. 
\end{align*}

\begin{proof}[\textbf{Proof of Theorem \ref{theorem:efficiency}}]
\textbf{Part (a)}
  Choose any arbitrary $w \in \W$ and $y \in \Y$.  
  Applying the quadratic expansion for 
  $
  \widetilde{\F}_{Y(w)}(y)
  = 
  \widehat{\F}_{Y(w)}^{simple}(y)
  -
  (\widehat{\P}_{X}^{(w)} - \widehat{\P}_{X})  
  \gamma_{y}^{(w)} 
  $, 
  we can show that 
  \begin{align}
    \Var
    \big(
      \widetilde{\F}_{Y(w)}(y)
    \big)
    =
    &
    \Var
    \big(
    \widehat{\F}_{Y(w)}^{simple}(y)
    \big)  
    -
    2
    \Cov
    \Big (
    \widehat{\F}_{Y(1)}^{simple}(y),
    (\widehat{\P}_{X}^{(w)} - \widehat{\P}_{X}) 
    \gamma_{y}^{(w)}
    \Big ) 
     +
    \Var
    \Big (
    (\widehat{\P}_{X}^{(w)} - \widehat{\P}_{X}) 
    \gamma_{y}^{(w)}
    \Big ).     \label{eq:var-1}
  \end{align}

  We can write 
  $\widehat{\P}_{X} = \sum_{w' \in \W} 
  \hat{\pi}_{w'} \widehat{\P}_{X}^{(w')}
  $. 
  It is assumed that 
  observations are a random sample and $n_{w'} /n = \pi_{w'} + o(1)$ for every $w' \in \W$
  as $n \to \infty$. 
  Furthermore, all unconditional and conditional functions are 
  bounded. 
  By applying the dominated convergence theorem, we can show 
  \begin{align}
    n \Cov
    \Big (
    \widehat{\F}_{Y(w)}^{simple}(y),
    (\widehat{\P}_{X}^{(w)} - \widehat{\P}_{X}) 
    \gamma_{y}^{(w)}
    \Big ) \notag
    &
      = n
    \Cov    \Big (
    \widehat{\F}_{Y(w)}^{simple}(y),
    (1 - \hat{\pi}_{w})\widehat{\P}_{X}^{(w)} 
    \gamma_{y}^{(w)}(X)
    \Big ) \notag \\
    &
      = 
    \frac{1 - \pi_{w}}{ \pi_{w}}
    \Cov
    \big (
      \1_{ \{Y(w) \leq y\} },
      \gamma_{y}^{(w)}(X)
    \big ) + o(1). 
    \label{eq:cov-1}
  \end{align}
   Similarly, we can show that  
  \begin{align}
    n \Var
    \big (
    (\widehat{\P}_{X}^{(w)} - \widehat{\P}_{X}) 
    \gamma_{y}^{(w)}
    \big ) \notag 
    & 
      =
    n \Var
    \big (
    (1 -   \hat{\pi}_{w})
    \widehat{\P}_{X}^{(w)} 
    \gamma_{y}^{(w)}
    \big )
    + 
    n
    \sum_{w' : w' \not = w }
     \Var
    \big (
    \hat{\pi}_{w'}
    \widehat{\P}_{X}^{(w')} 
    \gamma_{y}^{(w)}
    \big )
 \notag \\ 
    & 
      =
 \frac{  (1 - \pi_{w})^2}{ \pi_{w}} 
    \Var
    \big (
    \gamma_{y}^{(w)}(X)
    \big )
    +
      \sum_{w' : w' \not = w }
      \frac{    \pi_{w'}^{2}}{\pi_{w'}}
    \Var
    \big (
    \gamma_{y}^{(w)}(X)
    \big ) + o(1)
    \notag \\ 
    & 
      =
      \bigg (
           \frac{  (1 -   \pi_{w})^2}{ \pi_{w}} 
        +
      \sum_{w' : w' \not = w }
      \pi_{w'} 
      \bigg )
    \Var
    \big (
    \gamma_{y}^{(w)}(X)
    \big ) + o(1)\notag \\ 
    & 
  =
  \frac{1-\pi_{w}}{\pi_{w}}
    \Var
    \big (
    \gamma_{y}^{(w)}(X)
    \big ) + o(1). 
    \label{eq:cov-2}
  \end{align}
  It follows from 
  (\ref{eq:var-1})-(\ref{eq:cov-2}) that
  \begin{align}
    \label{eq:bb2}
    n
    \big \{
    \Var \big(
    \widehat{\F}_{Y(w)}^{simple}(y)
    \big)
    -
    \Var \big(
      \widetilde{\F}_{Y(w)}(y)
    \big)
      \big \}
    &= 
    \frac{1-\pi_{w}}{\pi_{w}}
    \big \{
     2
    \Cov
    \big (
      \1_{ \{Y(w) \leq y\} },
      \gamma_{y}^{(w)}(X)
    \big )
    -
    \Var
    \big (
    \gamma_{y}^{(w)}(X)
    \big )
    \big \} + o(1).
  \end{align}
  An application of the law of iterated expectation yields \begin{align*}
      \Cov
    \big (
      \1_{ \{Y(w) \leq y\} },
      \gamma_{y}^{(w)}(X)
    \big )
    =
    \Var \big (    E[  \1_{ \{Y(w) \leq y\} }|X]     \big ),
  \end{align*}
  which together with (\ref{eq:bb2}) shows  
   \begin{align*}
    n
    \big \{
    \Var \big(
    \widehat{\F}_{Y(w)}^{simple}(y)
    \big)
    -
    \Var \big(
      \widetilde{\F}_{Y(w)}(y)
    \big) 
      \big \}
    &= 
    \frac{1-\pi_{w}}{\pi_{w}}
    \Var
      \big (
      \gamma_{y}^{(w)}(X)       
      \big ) + o(1). 
  \end{align*}
  Since $\pi_{w} \in (0, 1)$ and
  $    \Var
      \big (
      \gamma_{y}^{(w)}(X)      
      \big )  \geq 0$, 
  it follows that 
  $
    \Var \big(
    \widehat{\F}_{Y(w)}^{simple}(y)
    \big)
    \geq 
      \Var \big(
      \widetilde{\F}_{Y(w)}(y)
    \big) + o(n^{-1}).
  $
  Here, the equality hold 
  only when 
  $F_{Y(w)|X}(y) = F_{Y(w)}(y)$
  or  
  $X$ has no predictive power for the event $\1_{ \{Y(w) \leq y\} }$.  

  \vspace{0.2cm}
    \textbf{Part (b)}
    Choose any arbitrary $y \in \Y$.  
    First, we shall show that, for any $w, w' \in \W$, 
    \begin{align}
    \label{eq:cov-A}
        n
        \Cov \big(
        \widehat{\F}_{Y(w)}(y), \widetilde{\F}_{Y(w')}(y) 
        \big)
        =
        \Cov \big(
        \gamma_{y}^{(w)}(X),
        \gamma_{y}^{(w')}(X)
        \big). 
    \end{align}
    Fix any two distinct treatment statuses $w, w' \in \W$. 
    We can write 
  $
  \widetilde{\F}_{Y(w)}(y)
  = 
  \big (\widehat{\F}_{Y(w)}^{simple}(y)
  - \widehat{\P}_{X}^{(w)}\gamma_{y}^{(w)}
  \big ) 
  +
  \widehat{\P}_{X} \gamma_{y}^{(w)} 
  $
  and also  
  $\widehat{\P}_{X}  \gamma_{y}^{(w)} = \sum_{v\in \W} 
  \hat{\pi}_{v} \widehat{\P}_{X}^{(v)}  \gamma_{y}^{(w)} 
  $.   
  Given random sample and the bi-linear property of the covariance function, 
  we can show that 
  \begin{align*}
        \Cov \Big(
        \widetilde{\F}_{Y(w)}(y), 
        \widetilde{\F}_{Y(w')}(y) 
        \Big)
        =&
        \Cov
        \Big(
        \widehat{\F}_{Y(w)}^{simple}(y)
        - \widehat{\P}_{X}^{(w)} \gamma_{y}^{(w)} ,
        \hat{\pi}_{w'}\widehat{\P}_{X}^{(w')}  \gamma_{y}^{(w')} 
        \Big ) 
        \\ 
        & + \Cov
        \big(
        \hat{\pi}_{w}\widehat{\P}_{X}^{(w)}  \gamma_{y}^{(w)}, 
        \widehat{\F}_{Y(w')}^{simple}
        - \widehat{\P}_{X}^{(w')} \gamma_{y}^{(w')}
        \big ) \\
        & + \Cov
        \big(
        \widehat{\P}_{X}  \gamma_{y}^{(w)}, 
        \widehat{\P}_{X}  \gamma_{y}^{(w')}
        \big ) 
    , 
    \end{align*}
    where it can be shown that the first and second terms on the right-hand side are equal zero, due to the fact that
    $     E \big [   \widehat{\F}_{Y(w)}^{simple}(y)
        - \widehat{\P}_{X}^{(w)}\gamma_{y}^{(w)}|X_{1}, \dots, X_{n}
        \big ] = 0
    $. 
    Furthermore, under the random sample assumption, we can show that 
    $ \Cov
        \big(
        \widehat{\P}_{X}  \gamma_{y}^{(w)}, 
        \widehat{\P}_{X}  \gamma_{y}^{(w')}
        \big ) =
        n^{-1}
      \Cov
        \big(
         \gamma_{y}^{(w)}(X), 
         \gamma_{y}^{(w')}(X)
        \big )   $. 
        Thus, we can prove the equality in (\ref{eq:cov-A}).

    Next, we compare the variance-covariance matrices of the simple and regression-adjusted estimators.
    By applying the result from part (a) of this theorem and the one in (\ref{eq:cov-A}), we are able to show that
    \begin{align*}
    & n 
    \big \{
     \Var \big (
    \widehat{\theta}_{y}^{simple}
    \big ) 
    -
    \Var \big (
    \widetilde{\theta}_{y}
    \big )  
    \big \} \\ 
    & \ \ \ \ \ \ 
    = 
    \left [
    \begin{array}{cccc}
     \frac{1-\pi_{1}}{\pi_{1}}
    \Var
      \big (
      \gamma_{y}^{(1)}(X)       
      \big ),
    & 
    -\Cov
        \big(
         \gamma_{y}^{(1)}(X), 
         \gamma_{y}^{(2)}(X)
        \big ),   
        & 
        \dots,  
        & 
        -\Cov
        \big(
         \gamma_{y}^{(1)}(X), 
         \gamma_{y}^{(K)}(X)
        \big )  \\
        -\Cov
        \big(
         \gamma_{y}^{(2)}(X), 
         \gamma_{y}^{(1)}(X)
        \big ),   
        &     
        \frac{1-\pi_{2}}{\pi_{2}}
    \Var
      \big (
      \gamma_{y}^{(2)}(X)    
      \big ), 
        & \dots,  
        & 
        -\Cov
        \big(
         F_{Y(2)|X}, 
         F_{Y(K)|X}
        \big ) 
        \\ 
          \vdots &\vdots&\ddots & \vdots\\ 
        -\Cov
        \big(
         \gamma_{y}^{(K)}(X), 
         \gamma_{y}^{(1)}(X)
        \big ),   
        &     
        -\Cov
        \big(
         \gamma_{y}^{(K)}(X), 
         \gamma_{y}^{(2)}(X)
        \big ), 
        & \dots,  
        & 
        \frac{1-\pi_{K}}{\pi_{K}}
    \Var
      \big (
      \gamma_{y}^{(K)}(X)       
      \big )
    \end{array}
    \right ]  + o(1),
 \end{align*}
which can be written as  
\begin{align*}
       n 
    \big \{
     \Var \big (
    \widehat{\theta}_{y}^{simple}
    \big ) 
    -
    \Var \big (
    \widetilde{\theta}_{y}
    \big )  
    \big \}
    = 
    E 
    \Big [
     \big (\gamma_{y}(X) - E[\gamma_{y}(X)]\big) 
     A
     \big(\gamma_{y}(X) - E[\gamma_{y}(X)]\big)^{\top} 
    \Big ] + o(1),
\end{align*}
where 
$\gamma_{y}(X) = [\gamma_{y}^{(1)}(X), \dots, \gamma_{y}^{(K)}(X)]^{\top}$ and  
\begin{align*}
    A:= 
        \left [
    \begin{array}{cccc}
    {\pi}_{1}^{-1} - 1,
    & 
    -1,   
        & 
        \dots,  
        & 
        -1 \\
        -1,   
        &     
    {\pi}_{2}^{-1} - 1,
        & \dots,  
        & 
        -1 
        \\ 
          \vdots &\vdots&\ddots & \vdots\\ 
        -1,   
        &     
        -1, 
        & \dots,  
        & 
    {\pi}_{K}^{-1} - 1
    \end{array}
    \right ] .
    \end{align*}
    The variant of 
    Lagrange's identity
    in 
    Lemma \ref{lemma:L-ind}
    with $\sum_{w \in \W} \pi_{w} =1$
    shows that, 
    for an arbitrary vector $v:=(v_{1}, \dots, v_{K})^{\top} \in \R^{k}$,  
    \begin{align*}
     & v^{\top}
     \big (\gamma_{y}(X) - E[\gamma_{y}(X)]\big) 
     A
     \big(\gamma_{y}(X) - E[\gamma_{y}(X)]\big)^{\top} 
     v \\ 
      & \ \ = 
     \sum_{w \in \W}
     \frac{
     v_{w}^{2}
     \big ( 
     \gamma_{y}^{(w)}(X)
     -
     E[\gamma_{y}^{(w)}(X)]
     \big )^{2}
     }{
        \pi_{w}
     } 
      - 
     \bigg (
     \sum_{w \in \W}
     v_{w}
     \big ( 
     \gamma_{y}^{(w)}(X)
     -
     E[\gamma_{y}^{(w)}(X)]
     \big )
     \bigg )^{2} \\ 
     & \ \ \  = 
     \frac{1}{2}
    \sum_{w \in \W} 
    \sum_{\substack{w' \in \W \\ w' \neq w}}
     \frac{
     \big \{ 
         v_{w}
         \big ( 
         \gamma_{y}^{(w)}(X)
         -
         E[\gamma_{y}^{(w)}(X)]
         \big )
         {\pi}_{w'}
         -
         v_{w'}
         \big ( 
         \gamma_{y}^{(w')}(X)
         -
         E[\gamma_{y}^{(w')}(X)]
         \big )
         {\pi}_{w}     \big \} ^{2}
     }{
     {\pi}_{w}
     {\pi}_{w'}
     } . 
    \end{align*}
    It follows that 
    \begin{align*} 
       v^{\top}
    \big \{
     \Var \big (
    \widehat{\theta}_{y}^{simple}
    \big ) 
    -
    \Var \big (
    \widetilde{\theta}_{y}
    \big )  
    \big \}
    v
    = 
    \frac{1}{2}
    \sum_{w \in \W} 
    \sum_{\substack{w' \in \W \\ w' \neq w}}
     \frac{
     \Var 
     \Big (
         v_{w}
         \gamma_{y}^{(w)}(X)
         {\pi}_{w'}
         -
         v_{w'}
         \gamma_{y}^{(w')}(X)
         {\pi}_{w}     
         \Big )
     }{
     {\pi}_{w}
     {\pi}_{w'}
     } 
    + o(n^{-1}). 
    \end{align*}
  The above equality implies the desired positive semi-definiteness result, because $ \Var 
     \big (
         v_{w}
         \gamma_{y}^{(w)}(X)
         {\pi}_{w'}
         -
         v_{w'}
         \gamma_{y}^{(w')}(X)
         {\pi}_{w}     
         \big )
         \geq  0$
         for any 
         $w, w' \in \W$ with $ w \neq w'$.
  
  Furthermore, the positive definite result holds when 
    $ \Var 
     \big (
         v_{w}
         \gamma_{y}^{(w)}(X)
         {\pi}_{w'}
         -
         v_{w'}
         \gamma_{y}^{(w')}(X)
         {\pi}_{w}     
         \big )
         > 0$
         for any $v \in \R^{k}$ with $v \neq 0$
         and 
         for any 
         $w, w' \in \W$ with $ w \neq w'$. 
         Because 
        $v \in \R^{k}$ is chosen arbitrarily except  $v \neq 0$
        and $\pi_{w} \in (0,1)$ for all $w \in \W$, 
        the condition for the positive definiteness can be written as  
         $
             \Var 
         \big (
         \gamma_{y}^{(w)}(X)
         -
         r
         \cdot 
         \gamma_{y}^{(w')}(X)
         \big )
         > 0
         $
         for any  $r \in \R$
         and 
         for any 
         $w, w' \in \W$ with $ w \neq w'$. 
\end{proof}

\newpage

\section{Multiplier Bootstrap Procedure} \label{app:multiplier-bootstrap} 
We can obtain pointwise and uniform confidence bands for distributional parameters using multiplier bootstrap following, for example, \citet{chernozhukov2013inference} and \citet{belloni2017program}. We outline the procedure to obtain uniform confidence bands in Algorithm \ref{alg:uniform-band}. The algorithm can be altered slightly to generate pointwise confidence bands as explained in Remark \ref{remark:pointwise-ci}.
\begin{algorithm}[!h]
   \caption{Multiplier bootstrap procedure to obtain uniform confidence bands}
   \label{alg:uniform-band}
\begin{algorithmic}
\STATE
\STATE {\bfseries Input:} Data $\{(X_i, W_i, Y_i)\}_{i=1}^{n}$; point estimates $\hat\theta_y$; influence functions $\hat \psi_y(Z_i):=\psi_y(Z_i; \hat\theta_y, \hat\gamma_y)$
\STATE 
\STATE
   \begin{enumerate}
    \item[(1)] Draw multipliers $\{\xi_i\}_{i=1}^{n}= \{m_{1,i}/\sqrt{2} + ((m_{2,i})^2-1)/2\}_{i=1}^{n}$ independently from the data $\{Z_i\}_{i=1}^{n}$, where $m_{1,i}$ and $m_{2,i}$ are i.i.d. draws from two independent standard normal random variables.
    \item[(2)] For each $y\in\mathcal Y$, obtain the bootstrap draws $\phi^b(\hat\theta_y)$ of $\phi(\hat\theta_y)$ as 
    \begin{equation*}
    \phi^b(\hat\theta_y) =\phi(\hat\theta_y^b) \text{ where } \hat\theta_y^b=\hat\theta_y + \frac{1}{n}\sum_{i=1}^{n} \xi_i \hat \psi_y(Z_i).
    \end{equation*}
    \item[(3)] Repeat (1)-(2) $B$ times and index the bootstrap draws by $b=1, \dots, B$.
    \item[(4)] Obtain bootstrap standard error estimates for $\phi(\hat\theta_y)$ for each $y\in\mathcal Y$ as 
    \begin{equation*}
        \hat \Sigma(y) = \frac{q_{0.75}(y)-q_{0.25}(y)}{z_{0.75}-z_{0.25}},
    \end{equation*}
    where $q_p(y)$ is the $p$-th quantile of $\{\phi^b(\hat\theta_y):1\leq b\leq B\}$ and $z_p$ is the $p$-th quantile of the standard normal distribution.
    
    \item[(5)] Construct the bootstrap draw of the Kolmogorov-Smirnov maximal t-statistic as
    \begin{equation*}
    t_{\max}^b= \max_{y\in\mathcal Y} \frac{|\phi^b(\hat\theta_y) - \phi(\hat\theta_y)|}{\hat \Sigma(y)},
    \end{equation*}
    \item[(6)] Obtain the bootstrap estimators of the critical values as 
    \begin{equation*}
        \hat t_{1-\alpha} = (1-\alpha)-\text{quantile of } \{t_{\max}^b: 1\leq b \leq B\}.
    \end{equation*}
    \item[(7)] Construct $(1-\alpha)\times 100\%$ uniform confidence band for $(\phi(\theta_y))_{y\in\mathcal Y}$ as 
    \begin{equation*}
        I^{1-\alpha} = \{[\phi(\hat\theta_y) \pm \hat t_{1-\alpha} \times \hat\Sigma(y)]: y\in\mathcal Y\}.
    \end{equation*} 
\end{enumerate}

\STATE {\bfseries Result:}   Uniform confidence band $I^{1-\alpha}$ for $(\phi(\theta_y))_{y\in\mathcal Y}$
\end{algorithmic}
\end{algorithm}

\begin{remark} \label{remark:pointwise-ci}
To obtain pointwise confidence intervals using bootstrap, we skip Steps (5)-(6) and use $z_{1-\alpha/2}$ as a critical value instead of $\hat t_{1-\alpha}$.
\end{remark}

\begin{remark}
Note that the multiplier bootstrap method is computationally efficient since it does not involve recomputing the nuisance estimates $\hat\gamma_y$ and the influence functions $\hat\psi_y$.    
\end{remark}

\section{Key Tools and Proofs for Section \ref{sec:asymptotic}} \label{app:asymptotics}
\subsection{Additional Definitions: Empirical Processes} \label{sec:empirical-process}
We first introduce some basic definitions related to empirical processes in this section. The following definitions are taken from \citet{van1996weak}.
\begin{definition}[Brownian bridge]
$\mathbb G_P$ is called a \emph{P-Brownian bridge} on $\mathcal F$ if it is a mean-zero Gaussian process with covariance function $E[\mathbb G_P f \mathbb G_P g]= P(fg)- P(f)P(g)$. 
\end{definition}
\begin{definition}[Covering numbers and entropies]
The \emph{covering number} $N(\epsilon,\mathcal{F}, \| \cdot \|)$ is the minimal number of balls $\{g: \|g-f\|<\epsilon\}$ of radius $\epsilon$ needed to cover the set $\mathcal F$. The centers of the balls need not belong to $\mathcal F$, but they should have finite norms. The \emph{entropy} is the logarithm of the covering number.
\end{definition}

\begin{definition}[Envelope function]
An \emph{envelope function} of a class $\mathcal F$ is any function $x \mapsto F(x)$ such that $|f(x)| \leq F(x)$ for every $x$ and $f$.
\end{definition}

\subsection{Key tools} \label{app-sub:asymptotics-key-tools}
In this subsection, we introduce some key tools to derive our asymptotic results. Let $ \{ Z_i\}_{i=1}^\infty$ be a sequence of i.i.d. copies of the random element $ Z$  taking values in the measure space $({\mathcal{Z}}, \mathcal{A}_{{\mathcal{Z}}})$ according to the probability law $P$ on that space. Let $\mathcal{F}_P= \{f_{t,P} : t \in T\}$ be a set  of suitably measurable functions $z \longmapsto f_{t,P}(z)$ mapping ${\mathcal{Z}}$ to $\mathbb{R}$, equipped with a measurable envelope $F_P: \mathcal{Z} \longmapsto \mathbb{R}$.  The class is indexed by $P \in \mathcal{P}$ and $t \in T$, where $T$ is a fixed, totally bounded semi-metric space equipped with a semi-metric $d_{T}$. Let $N(\epsilon,\mathcal{F}_P, \| \cdot \|_{Q,2})$ denote the $\epsilon$-covering number of the class of functions $\mathcal{F}_P$ with respect to the $L^{2}(Q)$ seminorm $\| \cdot \|_{Q,2}$ for $Q$ a finitely-discrete measure on $({\mathcal{Z}}, \mathcal{A}_{{\mathcal{Z}}})$.

\begin{theorem}[\textbf{Uniform in $P$ Donsker Property}]\label{lemma: uniform Donsker} Suppose that for $q > 2$
\begin{eqnarray}\label{eq: characteristics1}
&& \sup_{P \in \mP} \|F_P\|_{P,q} \leq C  \text{ and }  \  \lim_{\delta \searrow 0} \sup_{P \in \mathcal{P}} \sup_{ d_{T}(t, \bar t) \leq \delta} \| f_{t,P} - f_{\bar t,P}\|_{P,2} = 0.
\end{eqnarray}
Furthermore, suppose that
\begin{eqnarray}\label{eq: characteristics2}
\lim_{\delta \searrow 0} \sup_{P \in \mathcal{P}}\int_0^\delta \sup_Q \sqrt{ \log N(\epsilon \|F_P\|_{Q,2}, \mathcal{F}_P, \| \cdot \|_{Q,2})} d \epsilon =0.
\end{eqnarray}
Let $\mathbb{G}_P$ denote the P-Brownian Bridge, and  consider $$Z_{n,P} := (Z_{n,P}(t))_{t \in T} := (\Gn(f_{t,P}))_{t \in T},  \ \  Z_P := (Z_P(t))_{t \in T} := (\mathbb{G}_P(f_{t,P}))_{t \in T}.$$

\noindent (a) Then, $Z_{n,P} \rightsquigarrow Z_P$ in $\ell^\infty(T)$ uniformly in $P \in \mathcal{P}$, namely
$$
\sup_{P \in \mP} \sup_{h \in \textrm{BL}_1(\ell^{\infty}(T))} | \Ep^*_P h (Z_{n,P}) - \Ep_P h(Z_P) | \to 0.
$$
(b) The process $Z_{n,P}$ is stochastically equicontinuous  uniformly in $P \in \mathcal{P}$, i.e.,  for every $\varepsilon>0$,
$$
\lim_{\delta \searrow 0} \limsup_{n \to \infty} \sup_{P \in \mathcal{P}} \Pr^*_P\left (  \sup_{d_T  (t, \bar t) \leq \delta} | Z_{n,P}(t) - Z_{n,P}(\bar t)| > \varepsilon \right ) =0.
$$
(c) The limit process $Z_P$ has the following continuity properties:
$$
\sup_{P \in \mP} \Ep_P \sup_{t \in T} |Z_P(t) | < \infty,  \ \  \lim_{\delta \searrow 0} \sup_{P \in \mathcal{P}} \Ep_P \sup_{d_T(t, \bar t) \leq \delta} |Z_{P}(t) - Z_{P}(\bar t)| = 0.
$$
(d) The paths $t \longmapsto  Z_P(t)$ are a.s. uniformly continuous on $(T,d_T)$ under each $P \in \mathcal P$.
\end{theorem}

\begin{proof}
See the proof of Theorem B.1 in \citet{belloni2017program}.    
\end{proof}

\paragraph{Uniform in $P$ Validity of Multiplier Bootstrap} Let $(\xi_i)_{i=1}^n$ be i.i.d. multipliers  whose distribution does not depend on $P$,
such that $\Ep \xi = 0$,  $\Ep \xi^2 = 1$, and $\Ep |\xi|^q \leq C$ for $q >2$. Consider the multiplier empirical process:
$$Z^*_{n,P} := (Z^*_{n,P}(t))_{t \in T} := (\Gn(\xi f_{t,P}))_{t \in T} := \left (\frac{1}{\sqrt{n}} \sum_{i=1}^n \xi_i f_{t,P} (Z_i) \right)_{t \in T}.$$
Here $\mathbb{G}_n$ is taken to be an extended empirical processes defined by
the empirical measure that assigns mass $1/n$ to each point $(Z_i, \xi_i)$ for $i=1,...,n$.
Let $Z_P = (Z_P(t))_{t \in T} = (\mathbb{G}_P(f_{t,P}))_{t \in T}$ as defined in Theorem \ref{lemma: uniform Donsker}.

\begin{theorem}[\textbf{Uniform in $P$ Validity of Multiplier Bootstrap}]\label{lemma: uniform Donsker for bootstrap} Assume
 the conditions of Theorem \ref{lemma: uniform Donsker} hold.

(a) Then, the  following  unconditional convergence takes place,  $Z^*_{n,P} \rightsquigarrow Z_P$
 in $\ell^{\infty}(T)$ uniformly in $P \in \mathcal{P}$, namely
 $$
\sup_{P \in \mP} \sup_{h \in \textrm{BL}_1(\ell^{\infty}(T))} | \Ep^*_P h (Z^*_{n,P}) - \Ep_P h(Z_P) | \to 0.
$$

(b) The following conditional convergence takes place,  $Z^*_{n,P} \rightsquigarrow_B Z_P$
 in $\ell^{\infty}(T)$ uniformly in $P \in \mathcal{P}$, namely uniformly in $P \in \mathcal{P}$
$$
 \sup_{h \in \textrm{BL}_1(\ell^{\infty}(T))} | \Ep_{B_n} h (Z^*_{n,P}) - \Ep_P h(Z_P) | =  o_P^*(1),
$$
where $\Ep_{B_n}$ denotes the expectation over the multiplier weights $(\xi_{i})_{i=1}^n$ holding the data $(Z_i)_{i=1}^n$ fixed.
\end{theorem}

\begin{proof}
See Theorem B.2 of \citet{belloni2017program} for the proof.
\end{proof}

\begin{definition}[\textbf{Uniform Hadamard Tangential Differentiability}]\label{def:uhd}
Consider a map $\phi: \mathbb{D}_\phi  \longmapsto\mathbb{E}$, where the domain of the map $\D_\phi $
is a subset of a normed space $ \D$ and the range is a subset of the normed space $\mathbb{E}$. Let
$\mathbb{D%
}_{0}$ be a normed space, with $\mathbb{D}%
_{0} \subset\mathbb{D}$, and  $\mathbb{D}_\rho$ be a compact metric space, a subset of $\mathbb{D}_\phi$. The map $\phi: \mathbb{D}_\phi \longmapsto\mathbb{E}$ is called \textit{Hadamard-differentiable uniformly} in $\rho \in \mathbb{D}_\rho$ tangentially to $\mathbb{D}_0$ with derivative map $h \longmapsto \phi'_\rho(h)$, if
\begin{equation*}
\Big |\frac{\phi(\rho_n+ t_{n} h_{n}) - \phi(\rho_n)}{t_{n}} -\phi_{\rho}^{\prime}(h)\Big|  \to 0,\ \   \Big| \phi'_{\rho_n} (h_n) -  \phi_{\rho}^{\prime}(h) \Big|    \to 0 , \ \ \ n \rightarrow\infty,
\end{equation*}
for all convergent sequences $\rho_n  \to \rho$ in $\mathbb{D}_\rho$, $t_{n} \rightarrow0$ in $\mathbb{R}$, and $h_{n} \rightarrow h \in \D_0$ in $\D$ such that $\rho_n+ t_{n} h_{n} \in\mathbb{D}_{\phi}$ for every $n$.  As a part of
the definition, we require
that the derivative map $h \longmapsto \phi_{\rho}^{\prime}(h)$ from $\mathbb{D}_0$
to $\mathbb{E}$ is linear for each $\rho \in \D_\rho$.
\end{definition}

\begin{theorem}[\textbf{Functional delta-method uniformly in $P \in \mP$}]\label{thm: delta-method} Let $\phi: \D_{\phi} \subset\D \longmapsto%
\mathbb{E}$ be Hadamard-differentiable uniformly in $ \rho \in \D_{\rho} \subset \D_{\phi}$ tangentially to $\mathbb{D%
}_{0}$ with derivative map  $\phi_{\rho}^{\prime}$. Let $\hat \rho_{n,P}$ be a sequence of stochastic processes taking values in $\D_{\phi}$, where each $\hat \rho_{n,P}$ is an estimator of the parameter $\rho_P \in \D_{\rho}$. Suppose there exists a sequence of constants $r_{n} \rightarrow\infty$ such that $Z_{n,P}= r_{n} (\hat \rho_{n,P}- \rho_P) \rightsquigarrow Z_P$  in $\D$ uniformly in $P \in \mP_n$. The limit process $Z_P$ is separable and takes its values in $\D_{0}$ for all $P \in \mP = \cup_{n \geq n_0} \mP_n$, where $n_0$ is fixed.  Moreover, the set of stochastic processes $\{Z_P: P \in \mP\}$  is relatively compact in the topology of weak convergence in $\D_0$, that is, every sequence  in this set can be split into weakly convergent subsequences. Then, $r_{n}\left( \phi(\hat \rho_{n,P}) - \phi(\rho_P) \right) \rightsquigarrow\phi_{\rho_P}^{\prime}(Z_P)$ in $\mathbb{E}$
 uniformly in $P \in \mP_n$. If $(\rho, h) \longmapsto \phi_{\rho}^{'}(h)$ is defined and continuous on the whole of $\D_\rho \times \D$, then the sequence $r_{n}\left( \phi(\hat \rho_{n,P}) - \phi(\rho_P) \right) -
\phi_{\rho_P}^{\prime}\left( r_{n} (\hat \rho_{n,P}- \rho_P)\right) $  converges to zero in outer probability  uniformly in $P \in \mP_n$.  Moreover, the set of stochastic processes $\{\phi'_{\rho_P}(Z_P): P \in \mP\}$  is relatively compact in the topology of weak convergence in $\mathbb{E}$.
\end{theorem}

\begin{proof}
See Theorem B.3 of \citet{belloni2017program}  for the proof.
\end{proof}

Let $D_{n,P}= (W_{i,P})_{i=1}^n$ denote the data vector and $B_n= (\xi_i)_{i=1}^n$ be a vector of random variables used to generate bootstrap. Consider sequences of stochastic processes
$\hat \rho_{n,P}= \hat \rho_{n,P}(D_{n,P})$, where  $Z_{n,P}=r_n(\hat \rho_{n,P}- \rho_{P}) \rightsquigarrow Z_P$ in the normed space $\D$ uniformly in $P \in \mP_n$. Also consider the bootstrap stochastic process $Z^{*}_{n,P} = Z_{n,P}(D_{n,P}, B_{n})$ in  $%
\D$, where $Z_{n,P}$ is a measurable function of $B_n$ for each value of $D_n$.
Suppose that  $Z^{*}_{n,P}$ converges conditionally given $D_{n}$ in distribution to $Z_P$
uniformly in $P \in \mP_n$, namely that
$$
 \sup_{h \in BL_1(\D)} | \Ep_{B_n} [h(Z^*_{n,P})] -  \Ep_{P}  h(Z_P)| = o^*_P(1),
$$
uniformly in $P \in \mP_n$, where $\Ep_{B_n}$ denotes the expectation computed with respect
 to the law of $B_n$ holding the data $D_{n,P}$ fixed and $BL_1(\mathbb{D})$ denotes the space of functions mapping $\mathbb{D}$ to $[0,1]$ with Lipschitz norm at most 1. This is  denoted as $Z^{*}_{n,P}
\rightsquigarrow_{B} Z_P $ uniformly in $P \in \mP_n$. Finally, let $\hat \rho^{*}_{n,P} = \hat \rho_{n,P}+ Z^{*}_{n,P}/r_n $ denote the bootstrap or simulation
draw of $\hat \rho_{n,P}$.

\begin{theorem}[\textbf{Uniform in $P$ functional delta-method for bootstrap and other simulation methods}]
\label{theorem:delta-method-bootstrap}
Assume the conditions of Theorem \ref{thm: delta-method} hold.  Let $\hat \rho_{n,P}$ and $\hat \rho^{*}_{n,P}$ be maps as
indicated previously taking values in $\D_{\phi}$ such that $r_n%
(\hat \rho_{n,P}- \rho_P) \rightsquigarrow Z_P$ and  $r_n(\hat \rho^{*}_{n,P} - \hat \rho_{n, P}) \rightsquigarrow_{B} Z_P$ in $\D$ uniformly in $P \in \mP_n$. Then, $X^*_{n,P}=r_n(\phi(\hat \rho ^{*}_{n,P}) -
\phi(\hat \rho_{n,P})) \rightsquigarrow_{B} X_P= \phi_{\rho_P}^{\prime }(Z_P) $ uniformly in $P \in \mP_n$.
\end{theorem}

\begin{proof}
See Theorem B.4  of \citet{belloni2017program} for the proof.
\end{proof}

\begin{lemma}[\textbf{Maximal Inequality \cite{chernozhukov2014gaussian}}]
\label{lemma:CCK} Suppose that $F\geq \sup_{f \in \mathcal{F}}|f|$ is a measurable envelope
with $\| F\|_{P,q} < \infty$ for some $q \geq 2$.  Let $M = \max_{i\leq n} F(Z_i)$ and $\sigma^{2} > 0$ be any positive constant such that $\sup_{f \in \mF}  \| f \|_{P,2}^{2} \leq \sigma^{2} \leq \| F \|_{P,2}^{2}$. Suppose that there exist constants $a \geq e$ and $v \geq 1$ such that
$\log \sup_{Q} N(\epsilon \| F \|_{Q,2}, \mF,  \| \cdot \|_{Q,2}) \leq  v (\log a + \log(1/\epsilon)), \ 0 <  \epsilon \leq 1.
$
Then
\begin{equation*}
\Ep_P [ \| \bG_{n} \|_{\mF} ] \leq C  \left( \sqrt{v\sigma^{2} \log \left ( \frac{a \| F \|_{P,2}}{\sigma} \right ) } + \frac{v\| M \|_{\Pr_P, 2}}{\sqrt{n}} \log \left ( \frac{a \| F \|_{P,2}}{\sigma} \right ) \right),
\end{equation*}
where $C$ is an absolute constant.  Moreover, for every $t \geq 1$, with probability $> 1-t^{-q/2}$,
\begin{multline*}
\| \bG_{n} \|_{\mF} \leq (1+\alpha) \Ep_P [ \| \bG_{n} \|_{\mF} ] + C(q) \Big [ (\sigma + n^{-1/2} \| M \|_{\Pr_P,q}) \sqrt{t}
+  \alpha^{-1}  n^{-1/2} \| M \|_{\Pr_P,2}t \Big ], \ \forall \alpha > 0,
\end{multline*}
where $C(q) > 0$ is a constant depending only on $q$.  In particular, setting $a \geq n$ and $t = \log n$,
with probability $> 1- c(\log n)^{-1}$,
\begin{equation} \label{simple bound}
\| \bG_{n} \|_{\mF} \leq C(q,c) \left ( \sigma \sqrt{v \log \left ( \frac{a \| F \|_{P,2}}{\sigma} \right ) } + \frac{v
 \| M \|_{\Pr_P,q} } {\sqrt{n}}\log \left ( \frac{a \| F \|_{P,2}}{\sigma} \right ) \right),
\end{equation}
where $  \| M \|_{\Pr_P,q}  \leq n^{1/q} \| F\|_{P,q}$ and  $C(q,c) > 0$ is a constant depending only on $q$ and $c$.
\end{lemma}

\begin{proof}
See \citet{chernozhukov2014gaussian} for the proof.     
\end{proof}

\begin{lemma}[\textbf{{Algebra for Covering Entropies}}] \text{  } \label{lemma: andrews}
(1) Let $\F$ be a VC subgraph class with a finite VC index $k$ or any
other class whose entropy is bounded above by that of such a VC subgraph class, then
the covering entropy of $\mF$ obeys:
\begin{equation*}
 \sup_{Q} \log  N(\epsilon \|F\|_{Q,2}, \F,  \| \cdot \|_{Q,2}) \lesssim 1+ k \log (1/\epsilon)\vee 0
\newline
\end{equation*}
(2) For any measurable classes of functions $\F$ and $\F^{\prime
} $ mapping $\mathcal{Z}$ to $\mathbb{R}$,
$$\begin{array}{l}
\log N(\epsilon \Vert F+F^{\prime }\Vert _{Q,2},\F+\F^{\prime
}, \| \cdot \|_{Q,2})\leq \log   N\left(\mbox{$ \frac{\epsilon }{2}$}\Vert F\Vert _{Q,2},\F, \| \cdot \|_{Q,2}\right)
+ \log N\left( \mbox{$ \frac{\epsilon }{2}$}\Vert F^{\prime }\Vert _{Q,2},\F^{\prime
}, \| \cdot \|_{Q,2}\right), \\
\log  N(\epsilon \Vert F\cdot F^{\prime }\Vert _{Q,2},\F\cdot \F^{\prime
}, \| \cdot \|_{Q,2})\leq \log   N\left( \mbox{$ \frac{\epsilon }{2}$}\Vert F\Vert _{Q,2},\F, \| \cdot \|_{Q,2}\right)
+ \log N\left( \mbox{$ \frac{\epsilon }{2}$}\Vert F^{\prime }\Vert _{Q,2},\F^{\prime
}, \| \cdot \|_{Q,2}\right), \\
 N(\epsilon \Vert F\vee F^{\prime }\Vert _{Q,2},\F\cup \F^{\prime
}, \| \cdot \|_{Q,2})\leq   N\left(\epsilon\Vert F\Vert _{Q,2},\F, \| \cdot \|_{Q,2}\right)
+ N\left( \epsilon\Vert F^{\prime }\Vert _{Q,2},\F^{\prime
}, \| \cdot \|_{Q,2}\right).
\end{array}
$$

(3)  Given a measurable class $%
\mathcal{F}$ mapping $\mathcal{Z}$ to $\mathbb{R}$ and a random variable $\xi$ taking values in $\mathbb{R}$,
\begin{equation*}
 \log \sup_{Q} N(\epsilon \Vert |\xi|F\Vert _{Q,2},\xi\F, \| \cdot \|_{Q,2})\leq \log \sup_{Q} N\left(
\epsilon /2\Vert F\Vert _{Q,2},\F, \| \cdot \|_{Q,2}\right)
\end{equation*}
(4)  Given measurable classes $\F_j$ and envelopes $F_j$, $j=1,\ldots,k$, mapping $\mathcal{Z}$ to $\mathbb{R}$, a function $\phi:\mathbb{R}^k\to\mathbb{R}$ such that for $f_j,g_j\in\F_j$,
$ |\phi(f_1,\ldots,f_k) - \phi(g_1,\ldots,g_k) | \leq \sum_{j=1}^k L_j(x)|f_j(x)-g_j(x)|$, $L_j(x)\geq 0$, and fixed functions $\bar f_j \in \F_j$,  the class of functions $\mathcal{L}=\{\phi(f_1,\ldots,f_k)-\phi(\bar f_1,\ldots,\bar f_k):f_j \in\mathcal{F}_j, j=1,\ldots,k\}$ satisfies
\begin{equation*}
 \log \sup_Q N(\epsilon\|\sum_{j=1}^kL_jF_j\|_{Q,2},\mathcal{L}, \| \cdot \|_{Q,2})\leq \sum_{j=1}^k\log  \sup_Q  N\left(
\mbox{$\frac{\epsilon}{k}$}\|F_j\|_{Q,2},\F_j, \| \cdot \|_{Q,2}\right).
\end{equation*}
\end{lemma}

\begin{proof}
See \citet{andrews1994empirical} for the proofs of $(1)$ and $(2)$. $(3)$ follows from $(2)$. See  Lemma K.1 of \citet{belloni2017program} for the proof of $(4)$.
\end{proof}

\subsection{Regularity conditions for Theorem \ref{theorem:uniform-clt}} \label{app:regularity-conditions}
In this subsection, we lay out the regularity conditions for Theorem \ref{theorem:uniform-clt}. In what follows, let $\delta$, $c_0$, $c$, and $C$ denote some positive constants. Let $\Delta_n \searrow 0$, $\delta_n \searrow 0$, and $\tau_n \searrow 0$ be sequences of constants approaching zero from above at a speed at most polynomial in $n$.

\begin{assumption}[Moment condition problem]\label{ass: S1}
Uniformly for all $n \geq n_0$ and $P \in \mathcal{P}_n$, the following conditions hold:
(i) The true parameter value $\theta_y$ obeys \eqref{eq:moment_condition} and is interior relative to $\Theta_y \subset \Theta \subset \mathbb{R}^{K}$, namely there is a ball
of radius $\delta$ centered at $\theta_y$ contained in $\Theta_y$ for all $y\in \mathcal{Y}$, and $\Theta$ is compact.

(ii) For
$\nu := (\nu_k)_{k=1}^{2K} = (\theta, t)$,  each $w \in \W$ and $y\in \mathcal{Y}$,  the map $ \Theta_y \times \Gamma_y \ni  \nu \longmapsto \Ep_P[\psi_{y}^{(w)}(Z; \nu)]$ is twice continuously differentiable a.s. with derivatives obeying the integrability conditions specified in Assumption \ref{ass: S2}.

(iii) The following identifiability condition holds: $\|\Ep_P[\psi_y(Z, \theta, \gamma_y)])\| \geq 2^{-1} ( \|(\theta- \theta_y)\| \wedge c_0)\ \text{  for all } \theta \in \Theta_y$.
\end{assumption}

\begin{assumption}[Entropy and smoothness]\label{ass: S2}
The set $(\mathcal{Y}, d_{\mathcal{Y}})$ is a semi-metric space such that $\log N(\epsilon, \mathcal{Y},  d_{\mathcal{Y}}) \leq  C \log (\mathrm{e}/\epsilon) \vee 0$.  Let
$\alpha \in [1,2]$, and let $\alpha_1$ and $\alpha_2$ be some positive constants. Uniformly for all $n \geq n_0$ and $P \in  \mathcal{P}_n$,  the following conditions hold: 

(i) The set of functions $\mathcal{F}_0 = \{  \psi_{y}^{(w)}(Z; \theta_y, \gamma_y):  w  \in \W, y \in \mathcal{Y}\}$, viewed as functions of $Z$ is suitably measurable;  has an envelope function $F_0(Z)= \sup_{w\in \W, y \in \mathcal{Y}, \nu \in \theta_y\times\Gamma_y}|\psi_{y}^{(w)}(Z; \nu)|$ that is measurable with respect to $Z$ and obeys $\|F_0\|_{P, q} \leq C$, where $q\geq 4$ is a fixed constant; and has a uniform covering entropy obeying
$
\sup_Q  \log N(\epsilon \|F_0\|_{Q,2}, \mathcal{F}_0, \| \cdot \|_{Q,2}) \leq C  \log(\mathrm{e}/\epsilon) \vee 0.
$

(ii) For all $w \in \W$ and  $k,r \in [2K]$, and $\psi_{y}^{(w)}(Z) := \psi_{y}^{(w)}(Z; \theta_y, \gamma_y )$,
\begin{enumerate}
\item [(a)] $\sup_{y \in \mathcal{Y}, (\nu, \bar \nu) \in (\theta_y\times \Gamma_y)^2} \Ep_P[  ( \psi_{y}^{(w)}(Z; \nu) - \psi_{y}^{(w)}(Z; \bar \nu))^2] / \| \nu - \bar \nu\|^{\alpha}\leq C$, $P$-a.s.,
\item [(b)] $\sup_{d_\mathcal{Y} (y, \bar y) \leq \delta } \Ep_P[  ( \psi_{y}^{(w)}(Z) - \psi_{\bar{y}}^{(w)}(Z))^2] \leq C \delta^{ \alpha_1},$
\item [(c)]
 $\Ep_P  \sup_{y \in \mathcal{Y}, \nu \in \Theta_y\times \Gamma_y} |\partial_{\nu_r} \Ep_P \left [ \psi_{y}^{(w)}(Z; \nu)\right ]|^2  \leq C$,
\item[(d)] $\sup_{y \in \mathcal{Y}, \nu \in \Theta_y\times \Gamma_y} |\partial_{\nu_k} \partial_{\nu_r}  \Ep_P[\psi_{y}^{(w)}(Z; \nu)]| \leq  C,$ $P$-a.s.
\end{enumerate}
\end{assumption}

\begin{assumption}[Estimation of nuisance functions]\label{ass: AS}
The following conditions hold for each $n \geq n_0$ and all $P \in  \mathcal{P}_n$. The estimated functions $\hat \gamma_y = (\hat \gamma_y^{(w)})_{w=1}^{K} \in \mathcal{G}_{yn}$ with probability at least $1- \Delta_n$, where
$\mathcal{G}_{yn}$ is the set of measurable maps $x\longmapsto \gamma = (\gamma^{(w)})_{w=1}^{K}(x) \in \Gamma_y(x)$ such that
$$
\| \gamma^{(w)} - \gamma_y^{(w)}\|_{P,2} \leq \tau_n,  \quad \tau_n^2 \sqrt{n} \leq \delta_n,
$$
and whose complexity does not grow too quickly in the sense that
$\mathcal{F}_1 = \{  \psi_{y}^{(w)}(Z; \theta, \gamma): w\in \W, y\in \mathcal{Y}, \theta \in \Theta_y, \gamma \in \mathcal{G}_{yn} \}$
is suitably measurable and its uniform covering entropy obeys
 $$
\sup_Q  \log N(\epsilon \|F_1\|_{Q,2}, \mathcal{F}_1, \| \cdot \|_{Q,2}) \leq \log (e/\epsilon) \vee 0,
$$
where $F_1(Z)$ is an envelope for $\mathcal{F}_1$ which is measurable with respect to  $Z$ and satisfies $F_1(Z) \leq F_0(Z)$ for $F_0$  defined in Assumption \ref{ass: S2}.
\end{assumption}

\subsection{Proofs for Section \ref{sec:asymptotic}}
In this subsection, we state the proofs for Theorems \ref{theorem:uniform-clt} and \ref{thm:uniform-bootstrap}.

\begin{proof}[\textbf{Proof of Theorem \ref{theorem:uniform-clt}}]

\textbf{\textsc{Step 0.}} In the proof $a \lesssim b$ means that $a \leq A b$, where the constant
$A$ depends on the constants  in Assumptions \ref{ass: S2}, but not on $n$ once $n \geq n_0$, and not on $P \in \mathcal{P}_n$. In Step 1, we consider a sequence $P_n$ in $\mathcal{P}_n$, but for simplicity, we write  $P =P_n$ throughout the proof,  suppressing the index $n$.  Since the argument is asymptotic, we can  assume that $n \geq n_0$ in what follows.

Also, let
\begin{eqnarray}\label{eq: B}
 B(Z) & : = & \max_{ w \in \W, k \in [2K]} \sup_{\nu \in \Theta_y \times \Gamma_y, y \in \mathcal{Y}} \Big | \partial_{\nu_k} \Ep_P[\psi_{y}^{(w)}(Z; \nu)] \Big |.
\end{eqnarray}

\textbf{\textsc{Step 1.}} (A Preliminary Rate Result). In this step, we claim that  with probability $1- o(1)$, 
$$\sup_{y \in \mathcal{Y}}\| \hat \theta_y - \theta_y\| \lesssim \tau_n.$$

Since $\hat \theta_y$ is defined as a solution to the sample moment condition, we have
$$\| \En \psi_y(Z, \hat \theta_y, \hat\gamma_y)\| \leq \inf_{\theta \in \Theta_y}\| \En \psi_y(Z, \theta, \hat \gamma_y)\| + \epsilon_n \text{  for each $y \in \mathcal{Y}$, }$$ 
where $\epsilon_n=o(n^{-1/2})$. This implies
via triangle inequality that uniformly in $y \in \mathcal{Y}$ with probability $1-o(1)$
\begin{equation}\label{eq: rate proof}
\Big \| P [\psi_y(Z; \hat \theta_y, \gamma_y)] \Big \| \leq \epsilon_n + 2 I_1+ 2 I_2 \lesssim  \tau_n,
 \end{equation}
for $I_1$ and $I_2$ defined in Step 2 below.   The $\lesssim$  bound in  (\ref{eq: rate proof}) follows
 from Step 2 and from the assumption $\epsilon_n = o(n^{-1/2})$.
 Since by Assumption \ref{ass: S1} (iii), $2^{-1} ( \| (\hat \theta_y- \theta_y)\| \wedge c_0)$ does not exceed the left side of  (\ref{eq: rate proof}), we conclude that
$\sup_{y \in \mathcal{Y}}\| \hat \theta_y - \theta_y\|\lesssim \tau_n.$

\textbf{\textsc{Step 2.}} (Define and bound $I_1$ and $I_2$)  We claim that with probability $1- o(1)$:
\begin{eqnarray*}
I_1 & :=&  \sup_{\theta \in \Theta_y, y \in \mathcal{Y}} \Big \| \En \psi_y(Z; \theta, \hat \gamma_y) - \En \psi_y(Z; \theta, \gamma_y ) \Big \| \lesssim  \tau_n , \\
I_2 & := &   \sup_{\theta \in \Theta_y, y \in \mathcal{Y}}  \Big \| \En \psi_y(Z; \theta, \gamma_y) - P \psi_y(Z; \theta, \gamma_y ) \Big \| \lesssim  \tau_n.
\end{eqnarray*}
To establish this, we can bound $I_1 \leq 2I_{1a} + I_{1b}$ and $I_2 \leq I_{1a}$, where
with probability $1- o(1)$,
\begin{eqnarray*}
I_{1a} & :=&  \sup_{\theta \in \Theta_y, y \in \mathcal{Y}, \gamma \in \mathcal{G}_{yn} \cup \{ \gamma_y \}}  \Big\| \En \psi_y(Z; \theta, \gamma) - P \psi_y(Z; \theta, \gamma )  \Big \| \lesssim  \tau_n , \\
I_{1b} & := &   \sup_{\theta \in \Theta_y, y \in \mathcal{Y}, \gamma \in \mathcal{G}_{yn} \cup \{ \gamma_y \} }  \Big \| P \psi_y(Z; \theta, \gamma) - P \psi_y(Z; \theta, \gamma_y )  \Big \| \lesssim  \tau_n.
\end{eqnarray*}
These bounds in turn hold by the following arguments.

In order to bound $I_{1b}$, we employ Taylor's expansion and the triangle inequality. For $\bar \gamma(X, y, w, \theta)$ denoting a point on a line connecting vectors $\gamma(X)$ and $\gamma_y(X)$, and ${t_m}$ denoting the $m$th element of the vector $t$,
\begin{eqnarray*}
I_{1b} & \leq & \sum_{w=1}^{K} \sum_{m=1}^{K}  \sup_{\theta \in \Theta_y, y \in \mathcal{Y}, \gamma \in \mathcal{G}_{yn}} \Big  |  P\left [ \partial_{t_m} P \left [ \psi_{y}^{(w)}(Z, \theta, \bar \gamma(X, y, w, \theta))\right ] (\gamma^{(m)}(X) - \gamma_{y}^{(m)}(X))  \right]  \Big | \\
& \leq&K \cdot K \cdot  \|B\|_{P,2} \max_{y \in \Y, \gamma \in \mathcal{G}_{yn}, w \in \W} \| \gamma^{(w)} - \gamma_y^{(w)}\|_{P,2},
\end{eqnarray*}
where the last inequality holds by the definition of $B(Z)$ given earlier and H\"{o}lder's inequality.
 By  Assumption \ref{ass: S2}(ii)(c),  $\|B\|_{P,2}\leq C$, and  by Assumption \ref{ass: AS},  $\sup_{y \in \Y,  \gamma \in \mathcal{G}_{yn}, w\in \W} \| \gamma^{(w)} - \gamma_y^{(w)}\|_{P,2} \lesssim \tau_n$, hence we conclude that $I_{1b} \lesssim \tau_n$ since K is fixed.

In order to bound $I_{1a}$, we employ the maximal inequality of Lemma \ref{lemma:CCK}
to the class $$\mathcal{F}_{1} = \{ \psi_{y}^{(w)}(Z, \theta, \gamma):  w \in \W, y \in \mathcal{Y}, \theta \in \Theta_y, \gamma \in
\mathcal{G}_{yn} \cup \{ {\gamma}_{y}\}  \},$$
defined in Assumption \ref{ass: AS} and equipped with an envelope  $F_1 \leq F_0$, to conclude that with probability $1-o(1)$,
\begin{eqnarray*}
I_{1a} & \lesssim \tau_n.
\end{eqnarray*}
Here we use that $\log \sup_{Q} N(\epsilon \| F_1 \|_{Q,2}, \mF_1,  \| \cdot \|_{Q,2}) \leq \log (e/\epsilon) \vee 0$ by Assumption  \ref{ass: AS}; $\|F_0\|_{P, q} \leq C$ and $ \sup_{f \in \mathcal{F}_1} \| f\|^2_{P,2} \leq \sigma^2 \leq \| F_0 \|^2_{P,2}$ for $c \leq \sigma \leq C$ by Assumption \ref{ass: S2}(i).

\textbf{\textsc{Step 3.}} (Linearization)  By definition, we have
$$\sqrt{n} \|  \En \psi_y(Z; \hat \theta_y, \hat \gamma_y ) \| \leq \inf_{\theta \in \Theta_y} \sqrt{n} \|  \En \psi_y(Z; \theta, \hat \gamma_y ) \|+ \sqrt{n}\cdot \epsilon_n.$$
By Taylor's theorem, for all $y \in \mathcal{Y}$,
\begin{eqnarray*}
\sqrt{n}  \En \psi_y(Z; \hat \theta_y, \hat \gamma_y ) &=&
 \sqrt{n}  \En \psi_y(Z; \theta_y, \gamma_y)  \\
 &-& \sqrt{n} (\hat \theta_y - \theta_y) + \mathrm{D}_{u,0}(\hat \gamma_y -\gamma_y ) + II_1(y) + II_2(y),
\end{eqnarray*}
where the terms $II_1(y)$ and $II_2(y)$ are defined in Step 4 and $\mathrm{D}_{y,0}(\hat \gamma_y -\gamma_y )$ is treated in the next paragraph. Then, by the triangle inequality, for all $y \in \mathcal{Y}$ and Steps 4 and 5, we have
\begin{eqnarray*}
&& \left \| \sqrt{n}  \En \psi_y(Z; \theta_y, \gamma_y) - \sqrt{n} (\hat \theta_y - \theta_y) + \mathrm{D}_{y,0}(\hat \gamma_y -\gamma_y )\right \| \\
 & & \leq  \epsilon_n \sqrt{n}  + \sup_{y \in \mathcal {Y}} \Bigg  ( \inf_{\theta \in \Theta_y} \sqrt{n} \|  \En \psi_y(Z; \theta, \hat \gamma_y ) \| + \|II_1(y)\| + \|II_2(y)\|  \Bigg ) = o_P(1),
\end{eqnarray*}  where
the $o_P(1)$ bound
follows from Step 4,  $\epsilon_n \sqrt{n} = o(1)$  by assumption, and Step 5.

Moreover, by the orthogonality condition:
$$
\mathrm{D}_{y,0}(\hat \gamma_y- \gamma_y ) :=  \Bigg ( \sum_{m = 1}^{K}  \sqrt{n}  P\Big [ \partial_{t_m} P [\psi_{y}^{(w)}(Z; \theta_y, \gamma_y)] (\hat \gamma^{(m)}(X) - \gamma_y^{(m)}(X)) \Big ] \Bigg )_{w=1}^{K} = 0.
$$
Conclude using Assumption \ref{ass: S1} (iii) that
{\small $$
\sup_{y \in \mathcal{Y}}\left \| -\sqrt{n}  \En \psi_y(Z; \theta_y, \gamma_y) +  \sqrt{n} (\hat \theta_y - \theta_y) \right \| \leq o_P(1).
$$}\!

Furthermore, the empirical process $(\sqrt{n}  \En\psi_y(Z; \theta_y, \gamma_y))_{ y \in \mathcal{Y}}$ is equivalent to
an empirical process $\Gn$ indexed by $ \mathcal{F}_{P} := \Big \{\psi_{y}^{(w)} :  w \in \W,  y \in \mathcal{Y}  \Big \}_{},$
where $ \psi_{y}^{(w)}$ is the $w$-th element of  $\psi_y(Z; \theta_y, \gamma_y)$ and we make explicit the dependence of $\mF_P$ on $P$.

The conditions on $\mathcal{F}_0$ in Assumption \ref{ass: S2}(ii) imply that $\mathcal{F}_{P}$ has a uniformly well-behaved uniform covering entropy by Lemma \ref{lemma: andrews}, namely
$$
\sup_{P \in \mathcal{P} = \cup_{n \geq n_0} \mathcal{P}_n}\log \sup_{Q} N(\epsilon \| C F_0 \|_{Q,2}, \mF_P,  \| \cdot \|_{Q,2}) \lesssim \log (e/\epsilon) \vee 0,
$$
where $F_P = C F_0$ is an envelope for $\mF_P$ since $\sup_{f \in \mF_P} |f| \lesssim C F_0$ by Assumption \ref{ass: S2} (i).  The class $\mF_P$ is therefore Donsker uniformly in $P$ because $\sup_{P \in \mathcal{P}} \| F_P \|_{P,q} \leq C \sup_{P \in \mathcal{P}} \| F_0\|_{P,q}$ is bounded by Assumption \ref{ass: S2} (ii), and $\sup_{P \in \mathcal{P}} \| \psi_y - \psi_{\bar y} \|_{P,2} \to 0$ as $d_{\Y}(y, \bar y) \to 0$ by Assumption  \ref{ass: S2} (ii) (b).  Application of Theorem \ref{lemma: uniform Donsker} gives the results of the theorem.

\textbf{\textsc{Step 4.}} (Define and Bound $II_1(y)$ and $II_2(y)$).  Let  $II_1(y) := (II_{1w}(y))_{w=1}^{K}$ and  $II_2(y) = (II_{2w}(y))_{w=1}^{K}$, where{\small  \begin{eqnarray*}
&& II_{1w} (y) :=    \sum_{r,k = 1}^{2K} \sqrt{n} P \left [  \partial_{\nu_k} \partial_{\nu_r}  P[\psi_{y}^{(w)}(Z, \bar \nu_y(X,  w))] \{ \hat \nu_{yr}(X) - \nu_{yr}(X)\}\{ \hat \nu_{yk}(X) - \nu_{yk}(X)\} \right],\\
&& II_{2w} (y)  :=  \Gn(    \psi_{y}^{(w)}(Z, \hat \theta_y, \hat \gamma_y)- {\psi_{y}^{(w)}(Z, \theta_y, \gamma_y) )} ,
\end{eqnarray*}}
$\nu_y(X): =  ( \nu_{yk}(X))_{k=1}^{2K} : = (\theta_y^\top, \gamma_y(X)^\top)^\top,$   $\hat \nu_y(X): =  ( \hat \nu_{yk}(X))_{k=1}^{d_{\nu}}:=
(\hat \theta_y^\top, \hat \gamma_y^\top)^\top$, and $\bar\nu_y(X, w )$ is a vector on the line connecting $\nu_y(X)$ and $\hat \nu_y(X)$.

First, by Assumptions  \ref{ass: S2}(ii)(d) and \ref{ass: AS}, the claim of Step 1, and the H\"{o}lder inequality,
\begin{eqnarray*}
\max_{w \in \W} \sup_{y\in \mathcal{Y}}|II_{1w}(y)| & \leq&    \sup_{y \in \mathcal{Y}}\sum_{r,k = 1}^{2K} \sqrt{n} P \left [  C | \hat \nu_{yr}(X) - \nu_{yr}(X)| | \hat \nu_{yk}(X) - \nu_{yk}(X)| \right] \\
& \leq & C \sqrt{n} {K^2}   \max_{k \in {[2K]}} \sup_{y \in \mathcal{Y}}\| \hat \nu_{yk} - \nu_{yk} \|^2_{P,2} \lesssim_P \sqrt{n}  \tau_n^2 = o(1).
\end{eqnarray*}

Second, we have that with probability $1-o(1)$,
$
\max_{w\in\W} \sup_{y \in \mathcal{Y}}|II_{2w}(y)|  \lesssim  \sup_{f \in \mathcal{F}_2} | \Gn(f)|,
$
where, for $\Theta_{yn} := \{ \theta \in \Theta_y: \| \theta - \theta_y \| \leq C \tau_n \}$,
$$
\mathcal{F}_2 =  \Big \{ \psi_{y}^{(w)}(Z; \theta, \gamma) - \psi_{y}^{(w)}(Z; \theta_y, \gamma_y):  w \in\W, y\in \mathcal{Y},
\gamma \in \mathcal{G}_{yn}, \theta \in \Theta_{yn}  \Big \}_{}.$$
{Application of Lemma \ref{lemma:CCK} with an envelope $F_2 \lesssim F_0$ gives that with probability $1-o(1)$
\begin{eqnarray}
\sup_{f \in \mathcal{F}_2} | \Gn(f)| \lesssim  \tau_n^{\alpha/2}   +  n^{-1/2}n^{\frac{1}{q}},
\end{eqnarray}
since $ \sup_{f \in \mF_2} |f| \leq 2 \sup_{f \in \mF_1} |f| \leq 2  F_0$ by  Assumption \ref{ass: AS};   $\|F_0\|_{P,q} \leq C$ by Assumption \ref{ass: S2}(i);  $\log \sup_Q   N( \epsilon \|F_2\|_{Q,2},  \mathcal{F}_{2}, \|\cdot\|_{Q,2}) \lesssim (1 + \log (e/\epsilon))\vee 0$ by Lemma \ref{lemma: andrews} because $\mF_2 = \mF_1 - \mF_0$ for the $\mF_0$ and $\mF_1$ defined in Assumptions  \ref{ass: S2}(i)  and  \ref{ass: AS}; and $\sigma$  can be chosen so that $
 \sup_{f \in \mathcal{F}_2} \| f\|_{P,2} \leq \sigma  \lesssim  \tau^{\alpha/2}_n$.  Indeed,
 \begin{eqnarray*}
  \sup_{f \in \mathcal{F}_2} \| f\|^2_{P,2} & \leq &    \sup_{w \in\W, y\in \mathcal{Y}, \nu \in \Theta_{yn} \times  \mathcal{G}_{yn}}   P \left( P[  ( \psi_{y}^{(w)}(Z; \nu(X)) - \psi_{y}^{(w)}(Z, \nu_y(X)))^2]  \right) \\
 & \leq &  \sup_{y \in \mathcal{Y}, \nu \in \Theta_{yn} \times  \mathcal{G}_{yn}}  P  \left ( C  \| \nu(X) - \nu_y(X) \|^{\alpha} \right)  \\
 & = &  \sup_{y \in \mathcal{Y}, \nu \in \Theta_{yn} \times  \mathcal{G}_{yn}}  C \| \nu - \nu_y \|^{\alpha}_{P, \alpha}  \leq    \sup_{ y \in \mathcal{Y}, \nu \in \Theta_{yn} \times  \mathcal{G}_{yn}} C  \| \nu - \nu_u \|_{P, 2}^{\alpha}  \lesssim \tau_n^{\alpha},
 \end{eqnarray*}
where the first inequality
follows by the law of iterated expectations;  the second inequality follows  by Assumption \ref{ass: S2}(ii)(a);
and the last inequality follows from $\alpha \in [1,2]$ by Assumption \ref{ass: S2},  the monotonicity of the norm
$\|\cdot\|_{P,\alpha}$ in $\alpha \in [1, \infty]$, and  Assumption \ref{ass: AS}.} Conclude that with probability $1-o(1)$
\begin{eqnarray}
\max_{w \in \W} \sup_{y \in \mathcal{Y}}|II_{2w}(y)| \lesssim \tau^{\alpha/2}_n   +  n^{-1/2}  n^{\frac{1}{q}}  = o(1).
\end{eqnarray}

\textbf{\textsc{Step 5.}} In this step we show that
$
 \sup_{y\in \mathcal {Y}} \inf_{\theta \in \Theta_y} \sqrt{n} \|  \En \psi_y(Z; \theta, \hat \gamma_y ) \| =o_P(1).
$
We have that with probability $1- o(1)$
$$
\inf_{\theta \in \Theta_y} \sqrt{n} \|  \En \psi_y(Z; \theta, \hat \gamma_y) \| \leq \sqrt{n} \|  \En \psi_y(Z; \bar  \theta_y, \hat \gamma_y) \|,
$$
where $\bar\theta_y = \theta_y +\En \psi_y(Z, \theta_y, \gamma_y)$, since $\bar  \theta_y \in \Theta_y$ for all $y \in \mathcal{Y}$ with probability $1- o(1)$, and, in fact, $\sup_{y \in \mathcal{Y}} \| \bar  \theta_y - \theta_y \| =O_P( 1/\sqrt{n})$
by the last paragraph of Step 3.

Then, arguing similarly to Step 3 and 4, we can conclude that uniformly in $y \in \mathcal{Y}$:
{\small \begin{eqnarray*}
\sqrt{n} \|  \En \psi_y(Z; \bar  \theta_y, \hat \gamma_y ) \| & \leq &  \sqrt{n} \|  \En \psi_y(Z; \theta_y, \gamma_y ) -(\bar  \theta_y - \theta_y) + {\mathrm{D}_{y,0}(\hat \gamma_y - \gamma_y)} \| + o_P(1)
\end{eqnarray*}}\!
where the first term on the right side is zero by definition of $\bar  \theta_y$ and  $\mathrm{D}_{y,0}(\hat \gamma_y - \gamma_y) = 0$.
\end{proof}

\begin{proof}[\textbf{Proof of Theorem \ref{thm:uniform-bootstrap}}]
\textsc{Step 0.}  In the proof $a \lesssim b$ means that $a \leq A b$, where the constant
$A$ depends on the constants  in Assumptions  \ref{ass: S1}-- \ref{ass: AS}, but not on $n$ once $n \geq n_0$, and not on $P \in \mathcal{P}_n$. In Step 1, we consider a sequence $P_n$ in $\mathcal{P}_n$, but for simplicity, we write  $P =P_n$ throughout the proof,  suppressing the index $n$.  Since the argument is asymptotic, we can  assume that $n \geq n_0$ in what follows.

We first show that
$$ \hat Z^*_{n,P} \rightsquigarrow_B  Z_{P} \text{ in }  \ell^\infty(\Y)^{K},
\text{ uniformly in $P \in \mP_n$}.$$

In other words, we first show that the multiplier bootstrap provides a valid approximation to the large sample law of $\sqrt{n}(\hat \theta_y- \theta_y)_{y \in \mathcal{Y}}$. Let $\Pn$ denote the measure that puts mass $n^{-1}$ at the points $(\xi_i, Z_i)$ for $i=1,...,n$.
Let $\En$ denote the expectation with respect to this measure, so that
$\En f = n^{-1} \sum_{i=1}^n f(\xi_i, Z_i)$, and $\Gn$ denote the corresponding empirical process $\sqrt{n} ( \En - P)$,
i.e.
$$
\Gn f = \sqrt{n}(\En f - P f) = \frac{1}{\sqrt{n}} \sum_{i=1}^n \Bigg ( f(\xi_i, Z_i) - \int f(s, z) d P_\xi (s) dP (z) \Bigg).
$$

Recall that we define the bootstrap draw as:
\begin{align*}
Z^*_{n,P}:= \sqrt{n}(\hat \theta^*- \hat \theta)  =  \( \frac{1}{\sqrt{n}} \sum_{i=1}^n \xi_i \hat \psi_{y}(Z_i) \)_{y\in \mathcal{Y}}= \left (\Gn \xi \hat \psi_y \right)_{y \in \mathcal{Y}},
\end{align*}
where
$
 \hat \psi_y(Z) = \psi_y(Z, \hat \theta_y, \hat \gamma_y).
$

\textsc{Step 1.}( {Linearization})  In this step we establish that
\begin{equation}\label{REL1}
\ \  \zeta^*_{n,P}:  = Z^*_{n,P}-   G^*_{n,P} =   o_P(1) \ \  \text{ in $\D=\ell^\infty(\Y)^{K}$},
\end{equation}
where $G^*_{n, P} :=   (\mathbb{G}_{n}  \xi \bar \psi_{y})_{y \in \mathcal{Y}},$
and  $\bar \psi_y(Z) = \psi_y(Z; \theta_y, \gamma_y)$.

To show (\ref{REL1}),  we note that {with probability $1-\delta_n$,  $\hat \gamma_y \in \mathcal{G}_{yn},$ $\hat \theta_y \in \Theta_{yn} = \{ \theta \in \Theta_y: \| \theta - \theta_y \| \leq C \tau_n \}$,  so that
$
\| \zeta^*_{n,P} \|_{\D} \lesssim \sup_{f \in \mF_3} |\Gn[\xi f]|,
$
where
$$
\mathcal{F}_{3} = \Big \{ \tilde \psi_{y}^{(w)}(\bar \theta_y, \bar \gamma_y) - \bar \psi_{y}^{(w)} :
w \in \W, y \in \Y,  \bar \theta_y \in \Theta_{yn},  \bar \gamma_y \in \mathcal{G}_{yn}\Big\},
$$
where $\tilde \psi_{y}^{(w)}(\bar \theta_y, \bar \gamma_y))$ is the $j$-th element of $ \psi_y(Z; \bar \theta_y, \bar \gamma_y(X))$, and $\bar \psi_{y}^{(w)}$ is the $j$-th element of $\psi_y(Z;  \theta_y,  \gamma_y(X))$.
By the arguments similar to those employed in the proof of the previous theorem, $\mF_3$ obeys $$
\log \sup_Q   N( \epsilon \|F_3\|_{Q,2}, \mathcal{F}_{3}, \|\cdot\|_{Q,2}) \lesssim  (1 +\log (e/\epsilon)) \vee 0,$$
for an envelope  $F_3 \lesssim F_0 $.
By Lemma \ref{lemma: andrews}, multiplication of this class by $\xi$ does not change the entropy bound modulo  an absolute constant, namely
$$
\log \sup_Q   N( \epsilon \||\xi| F_3\|_{Q,2}, \xi \mathcal{F}_{3}, \|\cdot\|_{Q,2}) \lesssim (1 + \log (e/\epsilon))\vee 0.
$$
Also $\Ep[\exp(|\xi|)] < \infty$ implies $(\Ep [\max_{ i \leq n} |\xi_i|^2])^{1/2} \lesssim \log n$, so that, using independence
of $(\xi_i)_{i=1}^n$ from $(Z_i)_{i=1}^n$ and Assumption \ref{ass: S2}(i),
 $$\|  \max_{i \leq n } \xi_i F_0(Z_i)\|_{\Pr_P, 2}   \leq  \|   \max_{i \leq n } \xi_i \|_{\Pr_P, 2}
 \| \max_{i \leq n }F_0(Z_i)\|_{\Pr_P,2} \lesssim n^{1/q}  \log n.$$
Applying Lemma \ref{lemma:CCK},
\begin{eqnarray*}
&&  \sup_{f \in \xi \mathcal{F}_3} | \Gn (f) |  = O_{P} \(  \tau_n^{\alpha/2}  +  \frac{n^{1/q}  \log n }{\sqrt{n}} \)  = o_P(1),
\end{eqnarray*}
for $\sup_{f \in \xi \mathcal{F}_3}\|  f\|_{P,2}= \sup_{f \in  \mathcal{F}_3}\| f\|_{P,2} \lesssim \sigma_n \lesssim \tau_n^{\alpha/2}$, where the details of calculations are similar to those in the  proof of Theorem \ref{theorem:uniform-clt}.
Indeed, with probability $1 - o(\delta_n),$
 \begin{eqnarray*}
  \sup_{f \in \mathcal{F}_3} \| f\|^2_{P,2} & \lesssim & \sup_{w\in\W, y\in \mathcal{Y}, \nu \in \Theta_{yn} \times  \mathcal{G}_{yn}}   \ P \left( P[  ( \psi_{y}^{(w)}(Z, \nu(X)) - \psi_{y}^{(w)}(Z, \nu_y(X)))^2]  \right) \\
 & \lesssim &  \sup_{y \in \mathcal{Y}, \nu \in \Theta_{yn} \times  \mathcal{G}_{yn}}   \| \nu - \nu_y \|^{\alpha}_{P, \alpha}   \\
 & \lesssim &   \sup_{ y \in \mathcal{Y}, \nu \in \Theta_{yn} \times  \mathcal{G}_{yn}}   \| \nu - \nu_y \|_{P, 2}^{\alpha}  \\
 & \lesssim & \tau_n^{\alpha},
 \end{eqnarray*}
where the first inequality
follows from the triangle inequality and the law of iterated expectations;  the second inequality follows  by Assumption \ref{ass: S2}(ii)(a) and Assumption \ref{ass: S2}(i); the third inequality follows from $\alpha \in [1,2]$ by Assumption \ref{ass: S2},  the monotonicity of the norm
$\|\cdot\|_{P,\alpha}$ in $\alpha \in [1, \infty]$, and  Assumption \ref{ass: AS}; and the last inequality follows from $ \| \nu - \nu_y \|_{P, 2} \lesssim \tau_n$ by the definition of $ \Theta_{yn}$ and $ \mathcal{G}_{yn}$.} The equation (\ref{REL1}) follows.

\textsc{Step 2}. Here we are claiming that
$Z^*_{n,P} \rightsquigarrow_B  Z_{P}$  in $\D=\ell^\infty(\Y)^{K}$,  under any sequence $P =P_n \in \mP_n$, were $Z_{P} =   (\mathbb{G}_{P}  \bar \psi_{y})_{y\in \mathcal{Y}}$. By the triangle inequality and Step 1,
\begin{eqnarray*}
&& \sup_{h \in BL_1(\D) } \Big |   \Ep_{B_n}  h ( Z^*_{n,P} )  - \Ep_P h ( Z_{P})  \Big |\leq \sup_{h \in BL_1(\D) } \Big |   \Ep_{B_n}  h (G^*_{n,P} )  - \Ep_P h ( Z_{P})  \Big |
+   \Ep_{B_n} ( \|  \zeta^*_{n,P} \|_{\D} \wedge 2 ),
\end{eqnarray*}
where the first term  is $o^*_P(1)$, since $
G^*_{n,P} \rightsquigarrow_B  Z_{P}$  by Theorem \ref{lemma: uniform Donsker for bootstrap}, and the second term is $o_P(1)$  because  $ \|\zeta^*_{n,P}\|_{\D}  = o_P(1) $ implies that  $ \Ep_P ( \|  \zeta^*_{n,P} \|_{\D} \wedge 2 ) =
\Ep_P \Ep_{B_n} ( \|  \zeta^*_{n,P} \|_{\D} \wedge 2 ) \to 0$, which in turn implies that  $\Ep_{B_n} ( \|  \zeta^*_{n,P} \|_{\D} \wedge 2 ) = o_P(1)$ by the Markov inequality.
\end{proof}

\section{Additional Experimental Details and Results} \label{sec:app_experiment}
All experiments are carried out using R version 4.3.1 on a MacBook Pro with Apple M2 Max chip and 64GB memory. The code is available at \href{https://github.com/CyberAgentAILab/dte-ml-adjustment}{https://github.com/CyberAgentAILab/dte-ml-adjustment}. Additionally, we are in the process of developing a Python package that implements our proposed method.

\subsection{Simulation Study} \label{sec:simulation_details}
\subsubsection{Data Generating Process (DGP)} \label{sec:dgp} We fix the number of covariates $d_x$ as $d_x=100$ and the sample size $n$ to be in $\{500, 1000, 5000\}$. For each $i=1,\dots, n$, we generate $X_i=(X_{1i}, \dots, X_{100i})$ from $U_{100}((0,1)^{100})$,  a multivariate uniform distribution on $(0,1)$. Binary treatment variable $W_i$ follows Bernoulli distribution with success probability of $\rho=0.5$. A continuous outcome variable $Y_i$ is then generated from the outcome equation $Y_i=f(X_i, W_i)+U_i$, where the error term $U_i \sim N(0,1)$. We consider the functional form of 
\begin{equation} \label{eq:dgp-outcome}
f(X_i, W_i)= W_i + \sum_{j=1}^{100} \beta_j X_{ji} + \sum_{j=1}^{100} \gamma_j X_{ji}^2
\end{equation}
so that the outcome is nonlinear in covariates. We set
\begin{align*}
\beta_j = 
    \begin{cases}
        1 \text{ for } j\in \{1,\dots, 50\} \\
        0 \text{ for } j\in \{51,\dots, 100\},
    \end{cases}
\end{align*}
and 
\begin{align*}
\gamma_j = 
    \begin{cases}
        1 \text{ for } j\in \{1,\dots, 50\} \\
        0 \text{ for } j\in \{51,\dots, 100\}.
    \end{cases}
\end{align*}
In other words, the first 50 variables are relevant but the other 50 variables are irrelevant to the outcome variable. 

\subsubsection{Evaluation metrics} We evaluate the performance of our estimators using

\begin{enumerate}
    \item Bias ratio, computed as $100\%\times\frac{\frac{1}{R}\sum_{r=1}^{R}(\hat\Delta_{y,r} - \Delta_y)}{\Delta_y}$
    \item Root mean squared error (RMSE), computed as $RMSE=\sqrt{\frac{1}{R}\sum_{r=1}^{R}(\hat\Delta_{y,r} - \Delta_y)^2}$
    \item RMSE reduction, computed as $100\% \times (1-\frac{RMSE_{adjusted}}{RMSE_{simple}})$
\end{enumerate}
where $\Delta_y$ is the true distributional parameter (e.g., DTE, QTE) at threshold $y$, $\hat\Delta_{y,r}$ is the estimate from the replication $r$, and $R$ is the number of replications. We consider $R=1000$ in our experiments. To approximate the true distributions, we generate a dataset with 1,000,000 observations and calculate the distributional parameter at each $y$. In the simulations, we consider 9 values of threshold $y$ at the quantiles $\{0.1, 0.2, \dots, 0.9\}$ of the true outcome distribution.

\subsubsection{Implementation}
In the simulation study, we implement logistic LASSO with \textit{cv.glmnet} function in \texttt{glmnet} package in R \cite{friedman2010regularization, tay2023elastic}.

\subsubsection{Relevance of covariates and RMSE reduction}
We demonstrate the relationship between the predictive power of covariates on the outcome and the reduction in RMSE using a simple experiment. For this purpose, we explore a series of data generating processes where the relevance of covariates ranges from high to low. Specifically, we construct a slowly decaying sequence of coefficients $\kappa_s = 2\times s^{-1}$ for $s=1, \dots, 10$. Then, in our outcome equation \eqref{eq:dgp-outcome}, we set the coefficients as follows:
\begin{align*}
\beta_j = 
    \begin{cases}
        \kappa_s \text{ for } j\in \{1,\dots, 50\} \\
        0 \text{ for } j\in \{51,\dots, 100\},
    \end{cases}
\end{align*}
and 
\begin{align*}
\gamma_j = 
    \begin{cases}
        \kappa_s \text{ for } j\in \{1,\dots, 50\} \\
        0 \text{ for } j\in \{51,\dots, 100\}.
    \end{cases}
\end{align*}

Note that we consider a decaying sequence: $s=1$ corresponds to the case with the highest relevance, and the relevance diminishes as we increase $s$ up to $s=10$. Figure \ref{fig:rmse_relevance} illustrates how covariates with higher relevance result in a greater reduction in RMSE across all quantiles when sample size is $n=1000$.

\begin{figure*}[!htbp]
\vskip 0.2in
\begin{center}
\includegraphics[width=0.85\textwidth]{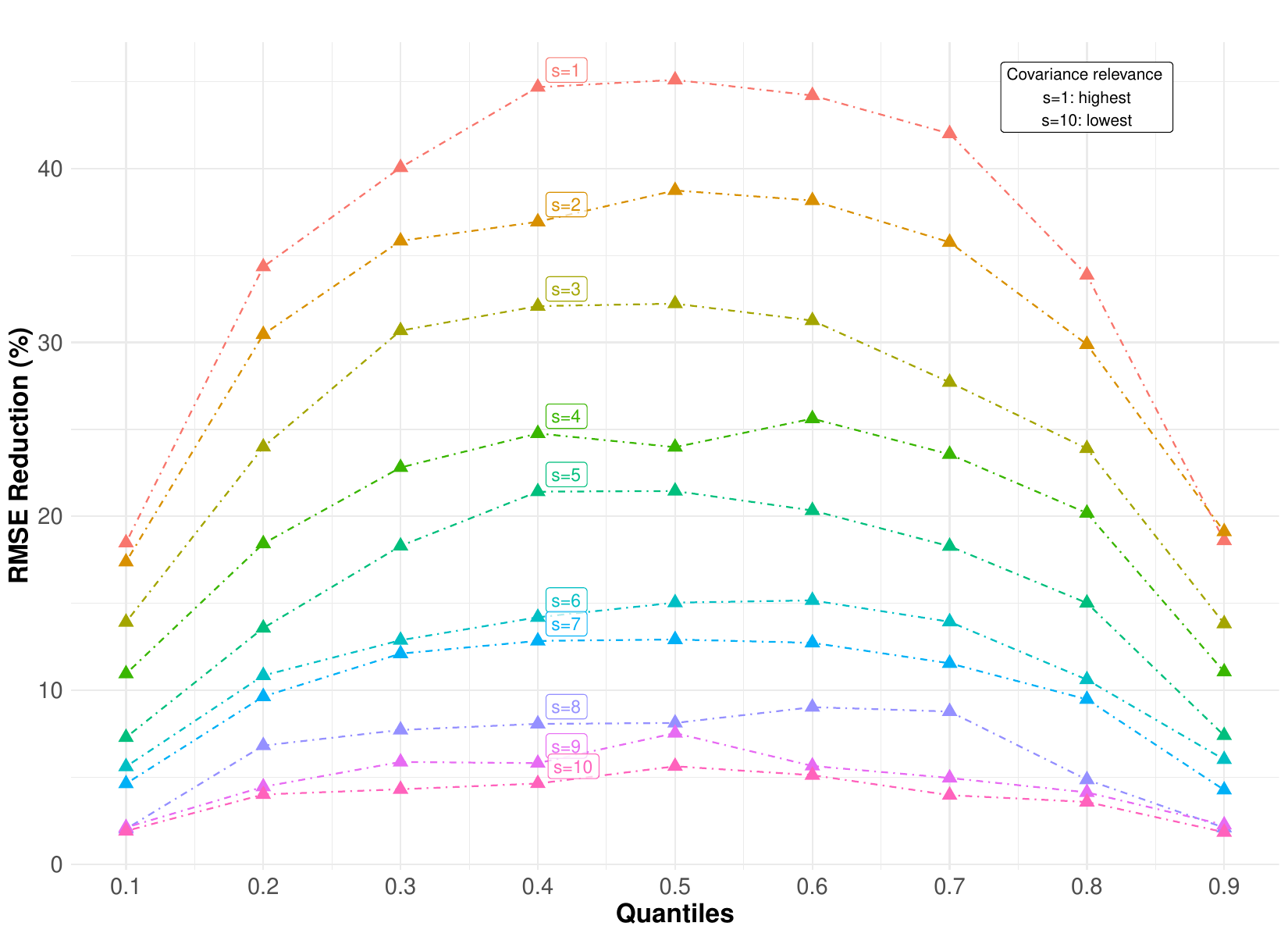}
\caption{RMSE reduction in \% of the ML adjusted estimator compared to the simple DTE estimator, under various data generating processes indexed by $s=1, \dots, 10$, calculated over 1000 simulations. $s=1$: highest relevance of covariates, diminishing relevance of covariates as $s$ increases up to $s=10$. The simple estimator is derived from empirical distribution functions, while the ML adjusted estimator is obtained using LASSO with 5-fold cross-fitting. $n=1000$.}
\label{fig:rmse_relevance}
\end{center}
\vskip -0.2in
    
\end{figure*}

\clearpage

\begin{figure*}[ht]
\vskip 0.2in
\begin{center}
\includegraphics[width=0.85\textwidth]{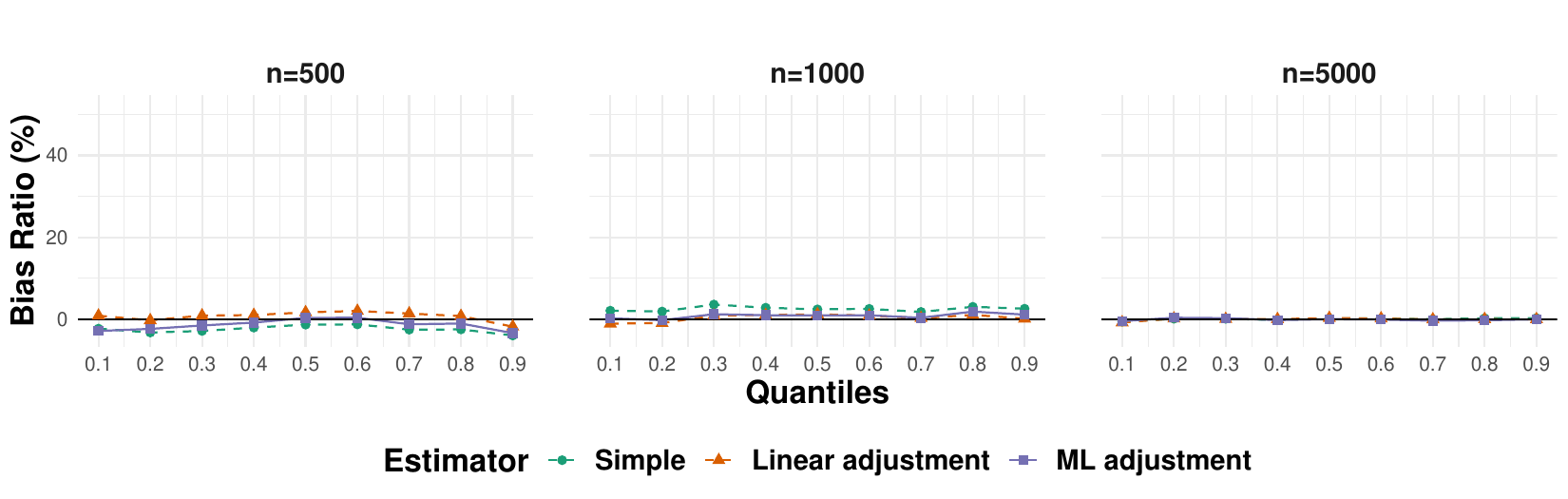}
\includegraphics[width=0.85\textwidth]{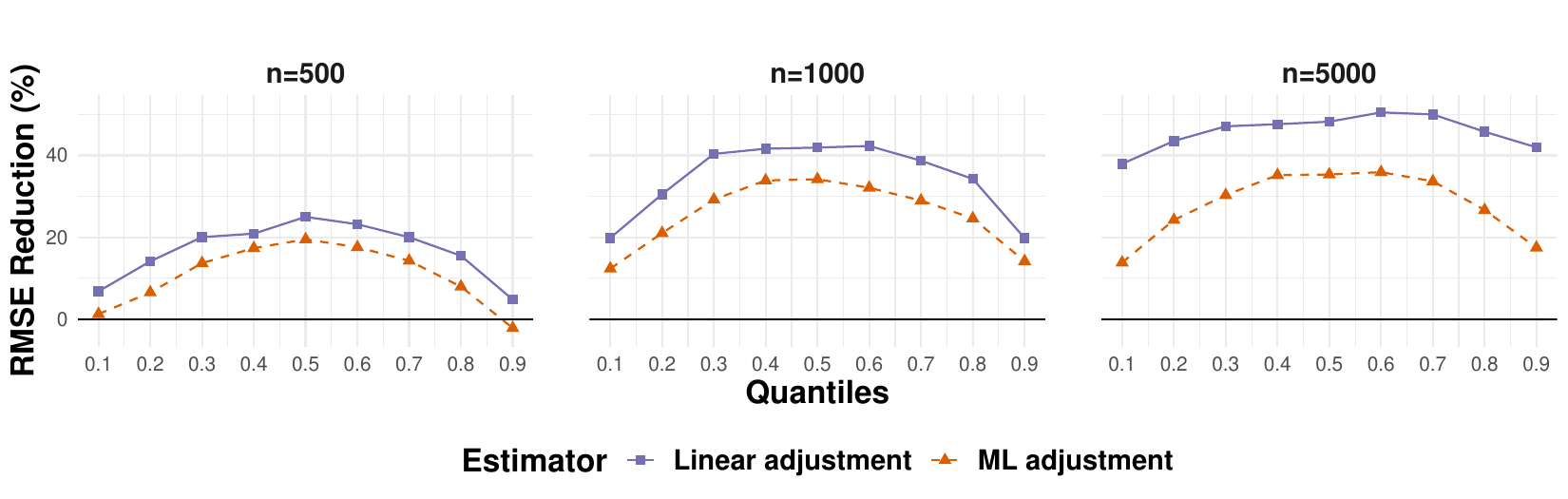}
\caption{Bias (top figure), as a \% of true value, of different QTE estimators and RMSE reduction in \% (bottom figure) of adjusted estimators compared to simple QTE estimator, under sample sizes $\{500, 1000, 5000\}$, calculated over 1000 simulations. The simple estimator is calculated from empirical distribution functions. The regression-adjusted estimators (linear adjustment and ML adjustment based on LASSO) are implemented using 5-fold cross-fitting.}
\label{fig:qte-bias-rmse}
\end{center}
\vskip -0.2in
\end{figure*} 

\subsubsection{Quantile Treatment Effect (QTE)} We also consider simple and regression-adjusted QTE estimators. In the setup introduced in Section \ref{sec:dgp}, the outcome variable is continuous and hence the QTE is well-defined. The true value of QTE is constant and equals 1 at all quantiles. The top figure of Figure \ref{fig:qte-bias-rmse} plots the bias as a \% of the true value of the QTE. The bottom figure of Figure \ref{fig:qte-bias-rmse} plots the RMSE reduction in \% terms for the linear and ML adjustment, compared to the simple estimator. We confirm the bias is small for all QTE estimators. Even when sample size is small ($n=500$), the bias is at most 4\%. As for the RMSE, the results are similar to that for the DTE explained in Section \ref{subsec:simulation}. The variance reduction is around 13\%-35\% for linearly adjusted estimator and is around 37\%-50\% for the ML adjusted estimator when sample size is large ($n=5000$).

\subsection{Nudges to reduce water consumption}\label{sec:water-nudge}
\subsubsection{Data and implementation } The dataset from the randomized experiment can be downloaded at \href{https://doi.org/10.7910/DVN1/22633}{https://doi.org/10.7910/DVN1/22633} \cite{DVN1/22633_2013}.

In our analysis, we implement gradient boosting with \texttt{xgboost} package in R \cite{Chen:2016:XST:2939672.2939785}.

\clearpage
\begin{figure*}[!htbp]
\vskip 0.2in
\begin{center}
\includegraphics[width=0.4\columnwidth]{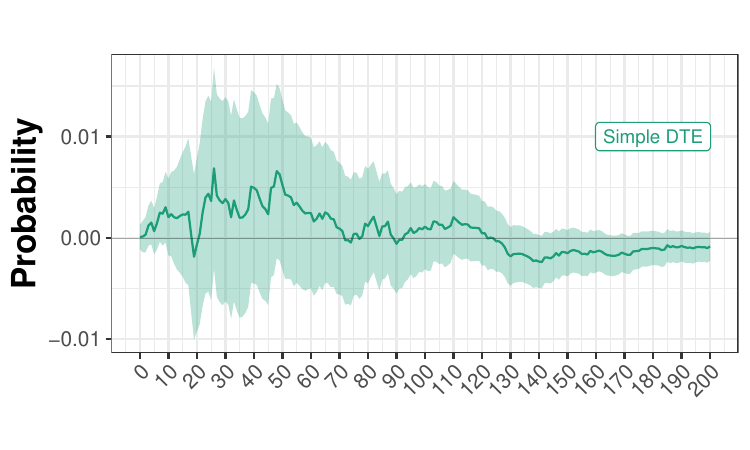}
\includegraphics[width=0.4\columnwidth]{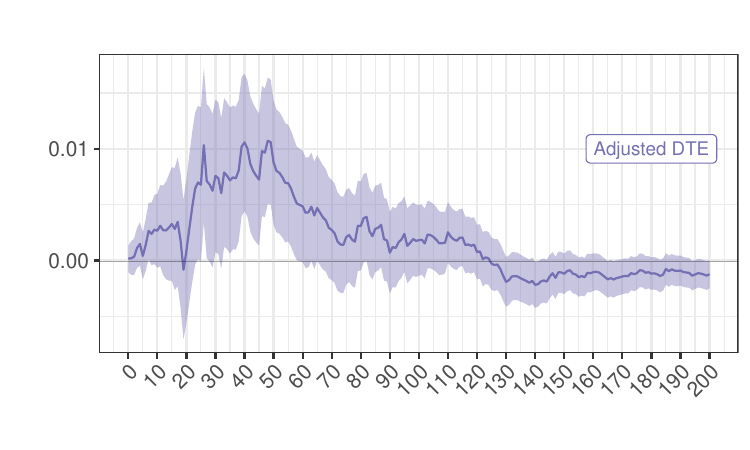}
\includegraphics[width=0.4\columnwidth]{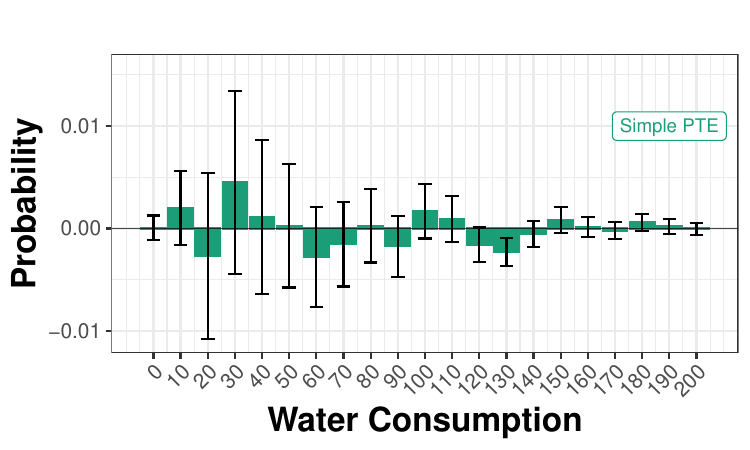}
\includegraphics[width=0.4\columnwidth]{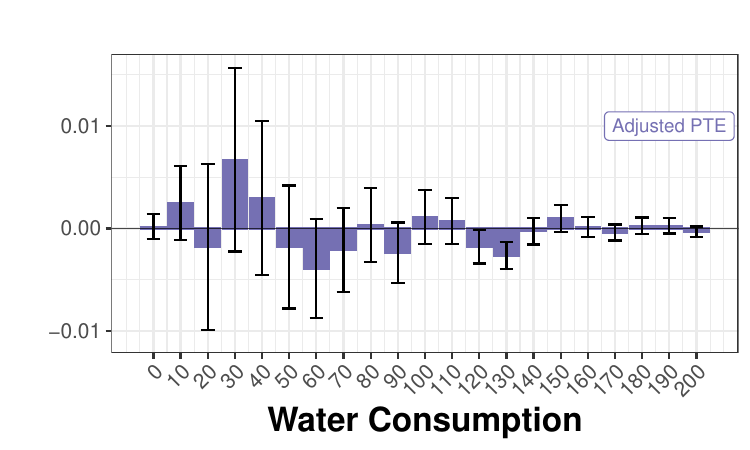}
\caption{Technical Advice (T1) vs. Control. Distributional Treatment Effect (DTE) and Probability Treatment Effect (PTE) on water consumption (in thousands of gallons). The top left figure represents the simple DTE; the top right figure depicts the regression-adjusted DTE, computed for $y\in\{0,1,2,\dots, 200\}$. The bottom left figure  represents the simple PTE; the bottom right figure represents the regression-adjusted PTE, computed for $y\in\{0,10,20,\dots, 200\}$ and $h=10$. The regression adjustment is implemented via gradient boosting with 10-fold cross-fitting. The shaded areas and error bars represent the 95\% pointwise confidence intervals. $n=78,478$.}
\label{fig:water-dte-t1}
\end{center}
\vskip -0.2in
\end{figure*}

\begin{figure*}[!htbp]
\vskip 0.2in
\begin{center}
\includegraphics[width=0.4\columnwidth]{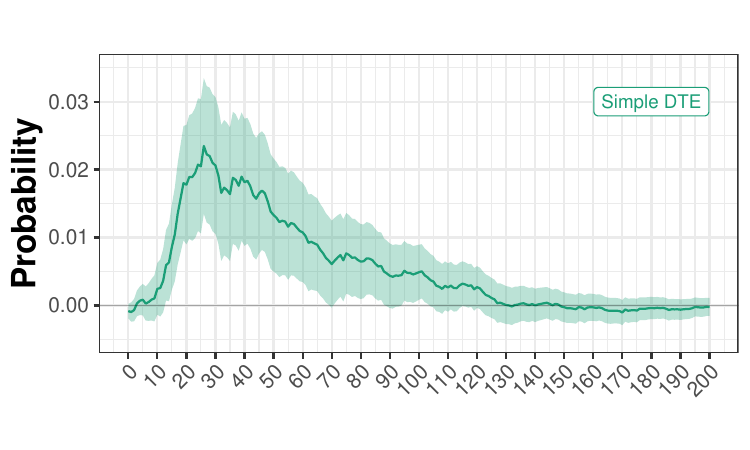}
\includegraphics[width=0.4\columnwidth]{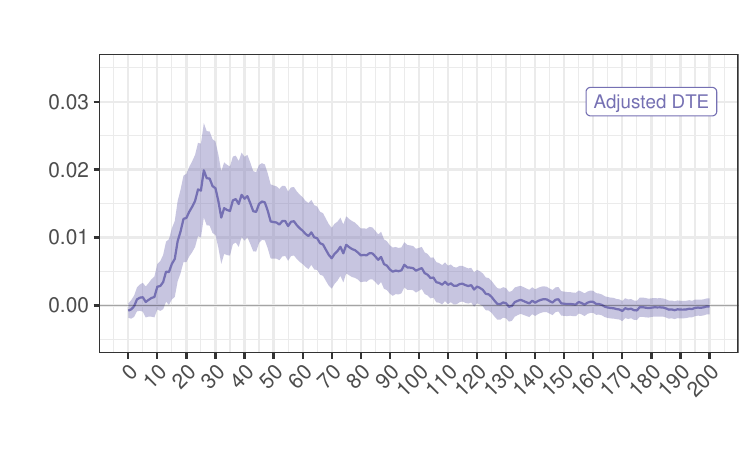}
\includegraphics[width=0.4\columnwidth]{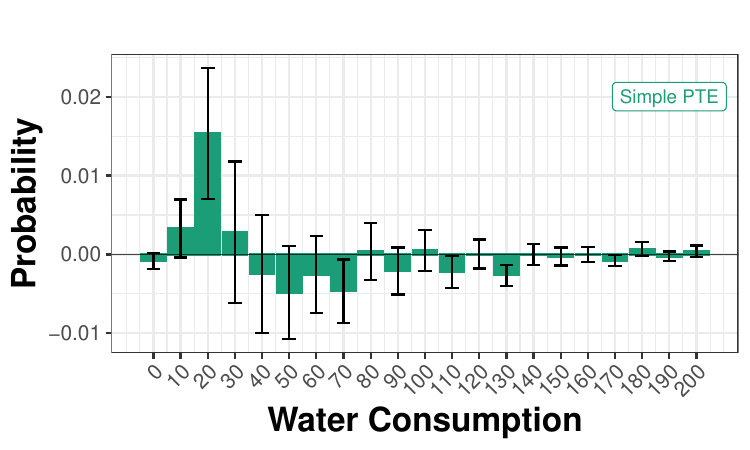}
\includegraphics[width=0.4\columnwidth]{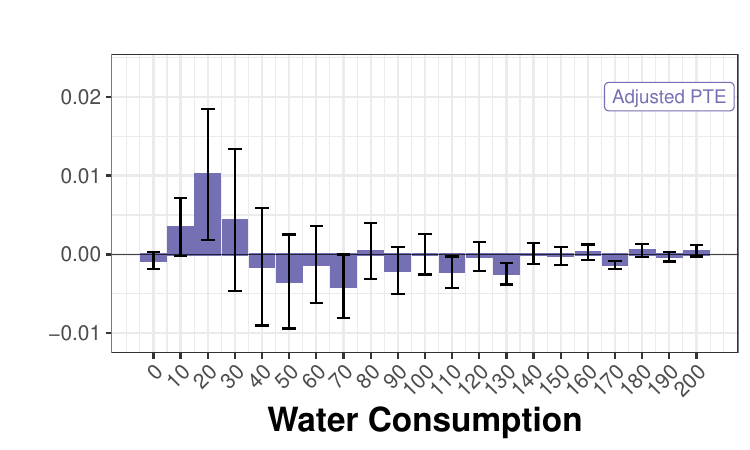}
\caption{Weak Social Norm (T2) vs. Control. Distributional Treatment Effect (DTE) and Probability Treatment Effect (PTE) on water consumption (in thousands of gallons). The top left figure represents the simple DTE; the top right figure depicts the regression-adjusted DTE, computed for $y\in\{0,1,2,\dots, 200\}$. The bottom left figure  represents the simple PTE; the bottom right figure represents the regression-adjusted PTE, computed for $y\in\{0,10,20,\dots, 200\}$ and $h=10$. The regression adjustment is implemented via gradient boosting with 10-fold cross-fitting. The shaded areas and error bars represent the 95\% pointwise confidence intervals. $n=78,468$.}
\label{fig:water-dte-t2}
\end{center}
\vskip -0.2in
\end{figure*}

\subsubsection{Results with multiple treatments}
The randomized experiment considered four treatment groups: technical advice (T1), weak social norm (T2), strong social norm (T3), and a control group. Technical advice (T1) involved providing residents with information on ways to reduce water use. The weak social norm (T2) treatment combined technical advice with an appeal to prosocial preferences. The strong social norm (T3) treatment further included social comparisons along with the elements of T2. On average, all three treatments resulted in a reduction in water use compared to the control group, with the strong social norm (T3) showing the largest effect and the technical advice (T1) showing the smallest effect. See \citet{ferraro2013using} for more details about the experimental design and average treatment effect analysis.

We extended the analysis by examining the entire distribution of water use. Figure \ref{fig:water-dte-t1} displays the DTE and PTE of technical advice (T1) compared to the control group. The PTE results indicate a reduction in water use in the range of (120, 140] under T1. Although regression adjustment results in tighter confidence intervals for the DTE and PTE, the overall conclusions remain the same.

Figure \ref{fig:water-dte-t2} presents the DTE and PTE of the weak social norm (T2) compared to the control group. Under T2, water use increased in the range of (20, 30] but decreased in the ranges of (110, 120], (130, 140], and (only slightly in) (170, 180]. Similar to T1, regression adjustment leads to tighter confidence intervals for the DTE and PTE without altering the primary conclusions for this treatment.

\end{document}